\documentclass{article}[11pt,a4paper]
\usepackage{outlines}
\usepackage{graphicx, subfigure, multirow, times, balance, color}
\usepackage{balance, url, amsfonts, verbatim, mathtools, paralist}
\usepackage{graphicx, xspace, subfigure}
\usepackage{comment}
\usepackage[letterpaper,hmargin={0.95in,0.95in},vmargin={0.9in,0.9in}]{geometry}
\usepackage[small]{caption}
\usepackage{textcomp}
\usepackage{hyperref}
\usepackage{clrscode3e}
\usepackage{multirow} 
\usepackage{algorithm}
\usepackage{algpseudocode}
\usepackage{amsthm}
\usepackage{titling}

\usepackage{mdwlist}

\newcommand{\bi}{\begin{itemize}}
\newcommand{\ei}{\end{itemize}}

\newcounter{counter}

\newcommand{\allnotes}[1]{}
\renewcommand{\allnotes}[1]{\textit{#1}}
\newcommand{\eat}[1]{}
\newcommand{\fixme}[1]{\allnotes{\bf\textcolor{red}{[#1]}}}

\newcommand{\notemarco}[1]{\allnotes{\textcolor{green}{[Marco: #1]}}}
\newcommand{\resilientproblem}{\textsc{Static-Routing-Resiliency}\xspace}

\newcommand{\nameprop}{shared-link-failure-free\xspace}
\newcommand{\nameproptitle}{Shared-link-failure-free\xspace}

\newcommand{\ucb}{UC Berkeley\xspace}

\newcommand{\hiit}{HIIT\xspace}

\newcommand{\huji}{HUJI\xspace}

\newcommand{\aalto}{Aalto University\xspace}

\newcommand{\dupalgo}{\textsc{Dup-Algo}\xspace}
\newcommand{\dupalgoOdd}{\textsc{Dup-Algo-Odd}\xspace}
\newcommand{\dfalgo}{\textsc{DF-Algo}\xspace}
\newcommand{\CircularRouting}{circular-arborescence\xspace}
\newcommand{\VertexCircular}{vertex-circular\xspace}
\newcommand{\Blue}{\texttt{Blue}}
\newcommand{\Orange}{\texttt{Orange}}
\newcommand{\Red}{\texttt{Red}}
\newcommand{\Green}{\texttt{Green}}

\newcommand{\duplication}{\texttt{DPL}\xspace}
\newcommand{\deterministic}{\texttt{DTM}\xspace}
\newcommand{\probabilistic}{\texttt{RND}\xspace}
\newcommand{\headerrewriting}{\texttt{HDR}\xspace}

\newtheorem{theorem}[counter]{Theorem}

\newtheorem{lemma}[counter]{Lemma}
\newtheorem{definition}{Definition}
\newtheorem{conjecture}[counter]{Conjecture}
\newtheorem{corollary}[counter]{Corollary}
\def\squarebox#1{\hbox to #1{\hfill\vbox to #1{\vfill}}}

\newcommand{\rephrase}[3]{\textsc{#1 #2}.~\emph{#3}}
\usepackage{latexsym}
\newcommand{\cT}{\mathcal{T}}
\newcommand{\cC}{\mathcal{C}}
\newcommand{\direct}[1]{\vec{#1}}
\newcommand{\HF}{H_F}
\newcommand{\dHF}{\direct{\HF}}
\newcommand{\dT}{\direct{T}}
\newcommand{\dG}{\direct{G}}
\newcommand{\Tinit}{T_{init}}
\newcommand{\Prob}[1]{\Pr\left[#1\right]}
\newcommand{\EE}[1]{{\mathbb E} \left[ #1 \right]}
\usepackage{xcolor}

\newcommand{\probAlgo}{\textsc{Bounced-Rand-Algo}\xspace}
\newcommand{\ADBED}{ADBED\xspace}
\newcommand{\interval}[1]{\left[#1 \right]_0}

\begin{document}
\onecolumn

\vspace{.1in}

\newcommand{\titletext}{Exploring the Limits of Static Resilient Routing}
\thispagestyle{empty}

\newcommand{\abstracttext}{


We present and study the \resilientproblem problem, motivated by routing on the Internet: Given a graph $G$, a unique destination vertex $d$, and an integer constant $c>0$, does there exist a static and destination-based routing scheme such that the correct delivery of packets from any source $s$ to the destination $d$ is guaranteed so long as (1) no more than $c$ edges fail and (2) there exists a physical path from $s$ to $d$? We embark upon a systematic exploration of this fundamental question in a variety of models (deterministic routing, randomized routing, with packet-duplication, with packet-header-rewriting) and present both positive and negative results that relate the edge-connectivity of a graph, i.e., the minimum number of edges whose deletion partitions $G$, to its resiliency. 
}

 
\eat{
}


\title{\titletext}

\author{
  Marco Chiesa\\
  HUJI
  \and
  Andrei Gurtov\\
  HIIT
  \and
  Aleksander M{\c a}dry\\
  MIT
  \and
  Slobodan Mitrovi\' c\\
  EPFL
  \and
  Ilya Nikolaevskiy\\
  Aalto University
  \and
  Aurojit Panda \\
  UC Berkeley
  \and
  Michael Schapira\\
  HUJI
  \and
  Scott Shenker\\
  ICSI/UC Berkeley
}

\eat{

\numberofauthors{6}

\author{
\alignauthor Marco Chiesa\\
    \affaddr{\huji}
\alignauthor Andrei Gurtov\\
       \affaddr{\hiit}
\alignauthor Aleksander Madry\\
   \affaddr{MIT}
\and  
\alignauthor Slobodan Mitrovic\\
       \affaddr{EPFL}
\alignauthor Ilya Nikolaevskiy\\
    \affaddr{\aalto}
\alignauthor Aurojit Panda\\
    \affaddr{\ucb}
\and  
\alignauthor Michael Schapira\\
       \affaddr{\huji}
\alignauthor Scott Shenker\\
       \affaddr{ICSI/\ucb}
}

}
\date{}

\maketitle  

\vspace{-.4in}
\begin{abstract}

{\em \abstracttext}
\end{abstract}

\vspace{-.1in}

\section{Introduction}
\label{sec:intro}

\vspace{-.1in}
\subsection{Motivation}

\vspace{-.05in}
Routing on the Internet (both within an organizational network and between such networks) typically involves computing a set of \emph{destination-based} routing tables (i.e., tables that map the destination IP address of a packet to an outgoing link). Whenever a link or node fails, routing tables are recomputed by invoking the routing protocol to run again (or having it run periodically, independent of failures). This produces well-formed routing tables, but results in relatively long outages after failures  
as the protocol is recomputing routes.

As critical applications began to rely on the Internet, such outages became unacceptable. As a result, ``fast failover'' techniques have been employed to facilitate immediate recovery from failures. The most well-known of these is Fast Reroute in MPLS where, upon a link failure, packets are sent along a precomputed alternate path without waiting for the global recomputation of routes~\cite{mplsfrr}. This, and other similar forms of fast failover thus enable rapid response to failures but are limited to the set of precomputed alternate paths.

The goal of this paper is to perform a systematic theoretical study of failover routing. The fundamental question is, how resilient can failover routing be? That is, how many link failures can failover routing schemes tolerate before connectivity is interrupted (i.e., packets are trapped in a forwarding loop, or hit a dead end)?  The answer to this question depends on both the structural properties of the graph, and the limitations imposed on the routing scheme. 

Clearly, if it is possible to store arbitrary amount of information in the packet header, perfect resiliency can be achieved by collecting information about every failed link that is hit by a packet~\cite{fcp,stephens-plinko-13}. Such approaches are not feasibly deployable in modern-day networks as the header of a packet may be too large for today's routing tables. Our focus is thus on failover routing schemes that involve only making minimal changes to packet headers, or even no changes at all. Another traditional approach to achieving high resiliency is implementing stateful routing, i.e., storing information at a node every time a packet is seen being received from a different incoming link (see, e.g., link reversal~\cite{gafni} and~\cite{lpsgss-ecdpm-13,ddc-hotnets}). As current routing protocols do not allow network operators to implement such stateful failover routing, our goal is to design protocols that correspond to a stateless, or \emph{static}, failover routing.

Specifically, we consider a particularly simple and practical form of static failover routing: for each incoming link, a router maintains a destination-based routing table that maps the destination IP address of a packet and the set of non-failed (``active'') links, to an output link. The router can locally detect which outgoing links are down and forwards packets accordingly. One should note that maintaining such per-incoming-link destination-based routing tables is necessary; not only is destination-based routing unable to achieve robustness against even a single link failure~\cite{kwong-link-protection-11}, but it is even computationally hard to devise failover routing schemes that maximize the number of nodes that are protected~\cite{rfc5286,borokhovic-shooting-13,kwong-link-protection-11,o2-03}. We only consider link failures, not router failures (which are not always detectable by neighboring routers, and so such fast failover techniques may not apply).

We now present positive and negative results for several models of interest, and end with an open conjecture.

\vspace{-.1in}
\subsection{Model(s) and Results}

\vspace{-.05in}
We now present an intuitive exposition of the failover routing models studied and our results. 

A failover routing algorithm is responsible for computing, for each node (vertex) of a network (graph), a {\em routing function} that {\em matches} an incoming packet to an outgoing edge. A set of routing functions for each vertex guarantees {\em reachability} between a pair of vertices, $u$ and $v$, for which there exists a connecting path in the graph, if any packet directed to node $v$ originated at node $u$ is correctly routed from $u$ to $v$. 

We are interested in routing functions that rely solely on information that is locally available at a node (e.g., the set of non-failed edges, the incoming link along which the packet arrived, and any information stored in the header of the packet). We consider four models of static failover routing: deterministic routing, randomized routing, routing with packet-duplication, and routing with (minimal) packet-header rewriting.

\vspace{0.1in}\noindent{\bf Deterministic (\deterministic) failover routing:} packets are matched on the destination label, the incoming edge, and the set of non-failed edges to an outgoing edge. Past work~\cite{uturn-06,keep-forwarding-14} (1) designed such functions with guaranteed robustness against \emph{only} a single link/node failure~\cite{enyedi-fast-reroute-07,podc,nelakuditi-fifr-07,wang-fifr-07,zhang-rpfp-13,zhang-fifr-05}, (2) achieved robustness against $\lfloor\frac{k}{2}-1 \rfloor$ edge failures for $k$-connected graphs~\cite{egr-ipfrmlf-14}, and (3) proved that it is impossible to be robust against any set of edge failures that does not partition the network~\cite{podc}.

We present the following positive  results for deterministic failover routing:

\begin{itemize}
\item For any $k$-connected graph, with $k\le 5$, one can find \deterministic routing functions that are robust to any $k-1$ failures.
\item For a variety of specialized $k$-connected graphs (including cliques, complete bipartite, hypercubes, Clos networks, hypercubes), one can find \deterministic routing functions that are robust to any $k-1$ failures.
\end{itemize}


Motivated by the possibility that one can protect against $k-1$ failures in \emph{some} $k$-connected graphs, we make the following general conjecture, whose proof eludes us despite much effort.
\begin{itemize}
\item 

{\bf Conjecture:} For any $k$-connected graph, one can find deterministic failover routing functions that are robust to any $k-1$ failures.
\end{itemize}

We present several negative results along these lines, e.g., for natural forms of deterministic failover routing. We show, in contrast, that slightly more expressive routing functions can indeed be robust to $(k-1)$ edge failures.


\vspace{0.1in}\noindent{\bf Randomized failover routing (\probabilistic):} as above, but the outgoing edge is chosen in a probabilistic manner. Observe that, in principle, in this model, even selecting an (active) outgoing edge uniformly at random achieves perfect resiliency. However, the expected delivery time of a packet, even if there was \emph{no} link failures, would be very large -- as large as $\Omega(mn)$ in some network topologies. Instead, we present a randomized protocol that guarantees the expected delivery time to be significantly improved and gracefully growing with the number of actual link failures.

\vspace{0.1in}\noindent{\bf Failover routing with packet-header rewriting (\headerrewriting):} a node has an ability to rewrite any bit of the packet header. Recent results showed that for any $k$-connected graph, $k$ bits are sufficient to compute routing functions that are robust to $(k-1)$ edge failures. We show that ability to modify at most \emph{three} bits suffices.

\vspace{0.1in}\noindent{\bf Failover routing with packet duplication (\duplication):} a node has an ability to duplicate a packet (without rewritting its header) and send the copies through deterministically chosen outgoing links. We show how to compute for any $k$-connected graph,  perfectly-resilient routing functions that do not create more than $k$ packets, where $k$ is the connectivity of the graph. (So, in particular, if there is no link failures, no packet duplication occurs.)

\vspace{-.1in}
\subsection{Organization}

\vspace{-.05in}
 In Section~\ref{sect:model}, we introduce our routing model and formally state the \resilientproblem problem. 
 In Section~\ref{sect:routing-technique}, we summarize our routing techniques that will be leveraged throughout the whole paper.
 In Section~\ref{sect:deterministic}, we present our main resiliency results for deterministic routing.
 In Section~\ref{sect:probabilistic}, we design an algorithm that, for any $k$-connected graph, computes randomized routing functions that are robust to $k-1$ edge failures and have bounded expected delivery time.
 In Section~\ref{sect:header-rewriting} and Section~\ref{sect:duplication}, we show that robustness to $(k-1)$ edge failures, where $k$ is the connectivity of a graph, can be achieved with deterministic routing function if just $3$ bits are added into the header of the packet and packet can be duplicated, respectively.
  In Section~\ref{sect:conclusion}, we draw our conclusions. Due to the lack of space, detailed proofs of each lemma and theorem can be found in the appendix section.

\vspace{-.1in}
\section{Model}\label{sect:model}

\vspace{-.05in}
 We represent our network as an undirected multigraph $G=(V(G),E(G))$, where each router in the network is modeled by a vertex in $V(G)$ and each link between two routers is modeled by an undirected edge in the multiset $E(G)$. When it is clear from the context, we simply write $V$ and $E$ instead of $V(G)$ and $E(G)$. We denote an (undirected) edge between $x$ and $y$ by $\{x,y\}$. 
 A graph is $k${\em-edge-connected}  if there exist $k$ edge-disjoint paths between any pair of vertices of $G$. 
 Each vertex $v$ routes packets according to a {\em routing function} that matches an incoming packet to a sequence of forwarding actions. 
 Packet {\em matching} is performed according to the set of active (non-failed) edges incident at $v$, the incoming edge, and any information stored in the packet header (e.g., destination label, extra bits), which are all information that are {\em locally} available at a vertex. 
 %
 Since our focus is on per-destination routing functions, we assume that there exists a unique destination  $d\in V$ to which every other vertex wishes to send packets and, therefore, that the destination label is not included is the header of a packet. 
 Forwarding {\em actions} consist in routing packets through an outgoing edge, rewriting some bits in the packet header, and creating duplicates of a packet.

 In this paper we consider four different types of routing functions. We first explore a particularly simple  routing function, which we call {\em deterministic} routing (\deterministic). In deterministic routing (Section~\ref{sect:deterministic}) a packet is forwarded to a specific outgoing edge based only on the incoming port and the set of active outgoing edges. 
 The other three routing functions, which are generalization of \deterministic are the following ones: {\em randomized} routing, in which a vertex forwards a packet through an outgoing edge with a certain probability, {\em header-rewriting} routing, in which a vertex rewrites the header of a packet, and {\em duplication} routing, in which a vertex creates copies of a packet. Deterministic routing is a special case of each of these routing functions.
 We present the formal definitions of the randomized, header-rewriting, and duplication routing models in Sections~\ref{sect:probabilistic},~\ref{sect:header-rewriting}, and~\ref{sect:duplication}, respectively. Observe that since deterministic, randomized, and duplication routing cannot modify a packet header, there is no benefit in matching it.

 \vspace{.1in}
 \noindent\textbf{The \resilientproblem (\textsc{Srr}) problem. }Given a graph $G$,  a  
  routing function $f$ is $k$\emph{-resilient} if, for each vertex $v\in V$, a packet originated at $v$ and routed according to $f$  reaches its destination $d$
  as long as at most $k$ edges fail and there still exists a path between $v$ and $d$. 
%
 The input of the \textsc{Srr} problem
 is a graph $G$, a destination $d\in V(G)$, and an integer $k>0$, and the goal is to compute a set of resilient routing functions that is $k$-resilient.
 

 \vspace{-.1in}
\section{General Routing Techniques}\label{sect:routing-technique}

\vspace{-.05in}
 \noindent\textbf{Definition and notation. }
 We denote a directed arc from $x$ to $y$ by $(x,y)$ and by $\direct{G}$  the directed copy of $G$, i.e. a directed graph such that $V(\direct{G})=V$ and $\{x,y\} \in E$ if and only if $(x,y),(y,x) \in E(\direct{G})$.
 \\
 A subgraph $T$ of $\direct{G}$ is an {\em $r$-rooted arborescence} of $\direct{G}$ if (i) $r \in V$, (ii) $V(T) \subseteq V$, (iii) $r$ is the only vertex without outgoing arcs and (iv), for each $v \in V(T)\setminus\{r\}$, there exists a single  directed path from $v$ to $r$ that only traverses vertices in $V(T)$.  
 If $V(T) = V$, we say that $T$ is a $r$-rooted {\em spanning} arborescence of $\direct{G}$. When it is clear from the context, we use the word ``arborescence'' to refer to a $d$-rooted spanning arborescence, where $d$ is the destination vertex.
 We say that two arborescences $T_1$ and $T_2$ are \emph{arc-disjoint} if $(x,y) \in E(T_1) \implies (x,y) \notin E(T_2)$.
 A set of $l$ arborescences $\{T_1, \ldots, T_l\}$ is {\em arc-disjoint} if the arborescences are pairwise arc-disjoint.
 We say that two arc-disjoint arborescences $T_1$ and $T_2$ do not \emph{share} an edge $\{x,y\} \in E$ if $(x,y) \in E(T_1) \implies (y,x) \notin E(T_2)$.\\
  \eat{\begin{figure}[tb]
  \centering
    \includegraphics[width=.15\columnwidth]{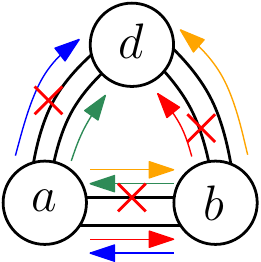}
    \caption{ A $4$-connected graph with $4$ arc-disjoint arborescences colored Blue, Red, Green, and Orange and $3$ failed edges
    depicted with a cross}
\label{fig:toy-gadget}
\end{figure}}
 As an example, consider Fig.\ref{fig:toy-gadget}, in which each pair of vertices is connected by two edges (ignore the red crosses) and four arc-disjoint ($d$-rooted spanning) arborescences $\Blue,\Orange,\Red$, and $\Green$ are depicted by colored arrows.

\vspace{.1in}
 \noindent\textbf{Arborescence-based routing. }
  Throughout the paper, unless specified otherwise, we let $\cT = \{T_1, \ldots, T_k\}$ denote a set of $k$ $d$-rooted arc-disjoint spanning arborescences of $\dG$. 
  All our routing techniques are based on a decomposition of $\dG$ into $\cT$. 
 The existence of $k$ arc-disjoint arborescences in any $k$-connected graph was proven in~\cite{e-edb-72}, while fast algorithms to compute such arborescences can be found in~\cite{bhalgat-feseacug-08}.
 We say that a packet is routed in {\em canonical} mode along an arborescence $T$ if a packet is routed through the unique directed path of $T$ towards the destination. 
%
%
 If packet hits a failed edge at vertex $v$ along $T$, it is processed by $v$ (e.g., duplication, header-rewriting) according to the capabilities of a specific routing function and it is rerouted along a different arborescence. We call such routing technique {\em arborescence-based} routing.
 One crucial decision that must be taken is the next arborescence to be used after a packet hits a failed edge. 
 In this paper, we propose two natural choices that represent the building blocks of all our routing functions.
 When a packet is routed along $T_i$ and it hits a failed arc $(v,u)$, we consider the following two possible actions: 
\begin{itemize}
 \item \textbf{Reroute along the next available arborescence}, e.g., reroute along $T_{next} = T_{(i+1)\mod k }$. 
 Observe that, if the outgoing arc belonging to $T_{next}$ is failed, we forward along the next arborescence, i.e. $T_{(i+2)\mod k }$, and so on.
 \end{itemize}
 \noindent
\begin{minipage}{\textwidth}
    \begin{minipage}{0.47\textwidth}
        \raggedright

        \begin{figure}[H]
          \centering
            \includegraphics[width=.25\columnwidth]{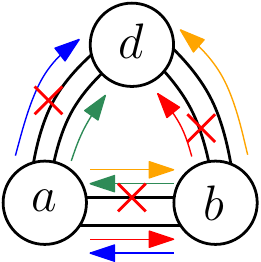}
            \caption{ A $4$-connected graph
            }
        \label{fig:toy-gadget}
        \end{figure}
        
        \vspace{-.2in}
    \begin{algorithm}[H]
        \caption{Definition of \probAlgo.}
        \label{algo:probabilistic-routing}
        \probAlgo: Given $\cT = \{T_1, \ldots, T_k\}$
                	\begin{enumerate}[1.]\addtolength{\itemsep}{-.4\baselineskip}
            	    \item $T :=$ an arborescence from $\cT$ sampled uniformly at random (u.a.r.)
            		\item While $d$ is not reached
            		    \begin{enumerate}[1.]\addtolength{\itemsep}{-.4\baselineskip}
            			    \item Route along $T$ (canonical mode)
            				\item If a failed edge is hit then
            				    \begin{enumerate}[(a)]\addtolength{\itemsep}{-.4\baselineskip}
            					    \item With probability $q$, replace $T$ by an arborescence from $\cT$ sampled u.a.r.
            			            \item Otherwise, bounce the failed edge and update $T$ correspondingly
                                \end{enumerate}
            	        \end{enumerate}
            	\end{enumerate}
            	
        \end{algorithm}

    \end{minipage}
    \hfill
    \noindent
    \begin{minipage}{0.47\textwidth}
    \raggedleft
        \begin{algorithm}[H]
        \caption{Definition of \dfalgo.}
        \label{algo:header-rewritin-routing}
                 
                 \vspace{.1in}
                 \dfalgo: Given $\cT = \{T_1, \ldots, T_k\}$ and $d$
                 \begin{enumerate}[1.]\addtolength{\itemsep}{-.4\baselineskip}
                    \item Set $i := 1$.
                    \item Repeat until the packet is delivered to $d$
                        \begin{enumerate}[1.]
                            \item Route along $T_i$ until $d$ is reached or the routing hits a failed edge.
                            \item If the routing hits a failed edge $a$ and $a$ is shared with arborescence $T_j$, $i \neq j$.
                            \begin{enumerate}[(a)]
                                \item Bounce and route along $T_j$.\footnote{As we discuss in the sequel, the routing scheme employed after bouncing might deviate from the one used before the bouncing has occurred. \vspace{.1in}}
                                \item If the routing hits a failed edge in $T_j$, route back to the edge $a$.
                            \end{enumerate}
                            \item Set $i := (i + 1) \mod k + 1$
                        \end{enumerate}
                 \end{enumerate}
        \vspace{.1in}
        \end{algorithm}
        
        \renewcommand\footnoterule{}

    \end{minipage}
\end{minipage}
\begin{itemize}
 \item \textbf{Bounce on the reversed arborescence}, i.e., we reroute along the arborescence $T_{next}$ 
 that contains arc $(u,v)$. 
\end{itemize}
 

 We say that a routing function is a {\em \CircularRouting} routing if each vertex can arbitrarily choose the first arborescence to route a packet and, for each $T_i\in\cT$, we use canonical routing until a packet hits a failed edge,  in which case  we reroute along the next available arborescence. 
 We will show an example in Section~\ref{sect:deterministic-arbitrary-graphs}.
 
 In the next sections, we show how it is possible to achieve different degrees of resiliency by using our general routing techniques and different routing functions (i.e., deterministic, randomized, packet-header-rewriting, and packet-duplication). 

\vspace{-.1in}
\section{Deterministic Routing}\label{sect:deterministic}

\vspace{-.05in}
In this section we show how to achieve $(k-1)$-resiliency for any arbitrary $k$-connected graph, with $k\le 5$, using deterministic routing functions (\deterministic), 
which map an incoming edge and the set of active edges incident at $v$ to an outgoing edge.
We show that for several $k$-connected graphs (e.g., cliques, hypercubes) there exists a set of $(k-1)$-resilient routing functions. In addition, we show that $2$-resiliency cannot be achieved for certain $2$-connected graphs. This motivate our conjecture: for any $k$-connected graph, does there exist a set of a $(k-1)$-resilient routing functions?

\vspace{-.1in}
\subsection{Arbitrary Graphs}\label{sect:deterministic-arbitrary-graphs}

\vspace{-.05in}
 We first show that  \CircularRouting routing is not sufficient to achieve $3$-resiliency.  Consider the example in  Fig.~\ref{fig:toy-gadget} with $3$ vertices $a$, $b$, and $c$ and $6$ edges (depicted as black lines) $e_{a,b}^{A}=\{a,b\}$, $e_{a,b}^{F}=\{a,b\}$, $e_{a,d}^{A}=\{a,d\}$, $e_{a,d}^{F}=\{a,d\}$, $e_{b,d}^{A}=\{b,d\}$, and $e_{b,d}^{F}=\{b,d\}$, where $A$ stands for ``active'' edge and $F$ for ``failed'' edge (depicted with a red cross over them). Four arc-disjoint arborescences $\cT=\{\Blue,\Orange,\Red,\Green\}$ are depicted by  colored arrows.
 Let $<\Blue,\Orange,\Red,\Green>$ be a circular ordering of the arborescences in $\cT$. We now describe how a packet $p$ originated at $a$ is forwarded throughout the graph using a \CircularRouting routing. Since $e_{a,d}^F$ is failed, $p$ cannot be routed along the \Blue arborescence. It is then rerouted through \Orange, which also contains a failed edge $e_{a,b}^F$ incident at $a$. As a consequence, $p$ is forwarded to $b$ through the \Red~arborescence. At this point, $p$ cannot be forwarded to $d$ because $e_{b,d}^F$, which belongs to \Red, failed. It is then rerouted through \Green, which also contains a failed edge $e_{a,b}^F$ incident at $b$. Hence, $p$ is rerouted again through \Blue, which leads $p$ to the initial state---a forwarding loop. 
 
 An intuitive explanation is the following one. Since an edge might be shared by two distinct arborescences, a packet may hit the same failed edge both when it is routed along the first arborescence and when it is routed along the second arborescence .
As a consequence, even $\frac{k}{2}$ failed edges may suffice to let a packet be rerouted along the same initial vertex and initial arborescence, creating a forwarding loop.
 Our first positive result shows that a forwarding loop cannot arise in $2$- and $3$-connected graphs if \CircularRouting routing is adopted. 

\newcommand{\TwoResiliency}{For any $k$-connected graph, with $k=2,3$, any \CircularRouting routing  is $(k-1)$-resilient. In addition, the number of switches between trees is at most $4$.}

\begin{theorem}\label{theo:2-resiliency}
 \TwoResiliency
\end{theorem}

 \begin{proof}[Proof sketch] Consider a $2$-connected graph $G=(V,E)$, two arc-disjoint arborescences $T_1$ and $T_2$ of $\vec G$, and an arbitrary failed edge $e=\{u,v\}\in E$. W.l.o.g, $T_1$ is the first arborescences that is used to route a packet $p$. When $p$ hits $e$ (w.l.o.g, at $u$), $p$ cannot hit $e$ in the opposite direction along $T_2$.
 In fact, this would mean that there exists a directed path from $u$ to $v$ that belongs to $T_2$ and that $(v,u)$ is contained in $T_2$---a directed cycle. A similar, but more involved argument, holds for the $3$-connected case (see Appendix~\ref{appe:deterministic-routing}).
 \end{proof}

\vspace{.1in}
\noindent\textbf{$4$-connected graphs. }
  Let us look again at the graph in Fig.~\ref{fig:toy-gadget}. It is not hard to see that a different circular ordering of the arborescences  (i.e.,  $<\texttt{Blue},\texttt{Green},\texttt{Orange},\texttt{Red}>$) would be robust to any three failures. 
 However, our first result shows that in general \CircularRouting routing is not sufficient to achieve $(k-1)$-resiliency, for any $k\ge4$. 
 
 \newcommand{\NoCircularRouting}{There exists a $4$-connected graph such that, given a set of $k$ arc-disjoint arborescences, there does not exist any $3$-resilient \CircularRouting routing function.}
 
 \begin{theorem}\label{theo:no-circular-routing}
  \NoCircularRouting
 \end{theorem}
 

To overcome this impossibility result, we first introduce the following lemma, in which we show how to construct four arc-disjoint arborescences such that some of them do not share edges with each other. Then, we compute a  \CircularRouting routing that is $3$-resilient based on these arborescences.

\newcommand{\SmartTreesTheorem}{
 For any $k$-connected graph $G$, with $k\ge1$, and any vertex $d \in V$, there exist $k$ arc-disjoint arborescences $T_1,\dots,T_{k}$ rooted at $d$ such that, if $k$ is even (odd), $T_1,\dots,T_{\frac{k}{2}}$ ($T_1,\dots,T_{\lfloor\frac{k}{2}\rfloor}$) do not share edges with each other and $T_{\frac{k}{2}+1},\dots,T_{k}$ ($T_{\lfloor\frac{k}{2}\rfloor+1},\dots,T_{k-1}$) do not share edges with each other.}

\begin{lemma}\label{lemm:bipartite-trees}
\SmartTreesTheorem
\end{lemma}

 The following theorem states that a circular ordering $<T_1,\dots,T_{4}>$ of the arborescences constructed as in Lemma~\ref{lemm:bipartite-trees} is a $3$-resilient \CircularRouting routing. We will make use of the general case of Lemma~\ref{lemm:bipartite-trees} in Sect.~\ref{sect:duplication}.

\newcommand{\ThreeResiliency}{For any $4$-connected graph, there exists a \CircularRouting routing that is $3$-resilient. In addition, the number of switches between trees is at most $2f$, where $f$ is the number of failed edges.}

\begin{theorem}\label{theo:3-resiliency}
\ThreeResiliency 
\end{theorem}

\vspace{.1in}
\noindent\textbf{$5$-connected graphs. }
We now leverage our second routing technique, i.e., bouncing a packet along the opposite arborescence when a packet hits a failed edge. The intuition behind bouncing a packet is the following one. When we bounce a packet along the opposite arborescence $T$, we know that at least one failed edge that belongs to $T$ is not contained in the path from $p$ to the destination vertex. 

Let $T_1,\dots,T_k$ be $k$ arc-disjoint arborescences of $\dG$ such that a \CircularRouting routing based on the first $k-1$ arborescences is $(c-1)$-resilient, with $c<k$. Let $R$ be a set of routing functions such that: each vertex that originates a packet $p$, forwards it along $T_k$ and, if a failed edge is hit along $T_k$, then $p$ is routed according to the \CircularRouting based on the first $k-1$ arborescences. Then, we have the following result.

\newcommand{\FourResiliencyLemma}{The set of routing functions $R$ is $c$-resilient.}

\begin{lemma}\label{lemm:plus-one-resiliency}
\FourResiliencyLemma
\end{lemma}
The $4$-resiliency for any $5$-connected graph now easily follows from Lemma~\ref{lemm:plus-one-resiliency} and Theorem~\ref{theo:3-resiliency}. We also show:
 
\newcommand{\FourResiliency}{For any $5$-connected graph $G$ there exist a set of $4$-resilient routing functions. In addition, the number of switches between trees is at most $2f$, where $f$ is the number of failed edges.}

\begin{theorem}\label{theo:4-resiliency}
\FourResiliency 
\end{theorem}

\begin{theorem}
For any $k$-connected graph there exist a set of $\lfloor\frac{k}{2}\rfloor$-resilient routing functions.
\end{theorem}
\begin{proof}
It easily follows from Lemma~\ref{lemm:plus-one-resiliency} and the fact that every \CircularRouting routing is $(\lfloor\frac{k}{2}\rfloor-1)$-resilient.
\end{proof}

Since every planar graph with no parallel edges is at most $5$-connected~\cite{diestel}, the following corollary easily follows.

\begin{corollary}
For any $k$-connected planar graph with no parallel edges there exist a set of $(k-1)$-resilient routing functions.
\end{corollary}


\vspace{.1in}
\noindent\textbf{Constrained topologies. }
For several graph topologies that are common in Internet routing or datacenter networks, we show that $(k-1)$-resilient routing functions can be computed in polynomial time. The list of graphs that
admit $(k-1)$-resilient routing functions encompasses cliques, complete bipartite graphs, generalized hypercubes, Clos networks, and grids~\cite{fattree,diestel,bcube-09}. We refer the reader to Appendix~\ref{appe:specific-networks} for further details. 

\newcommand{\CliqueResiliency}{For any $k$-connected clique graph there exist a set of $(k-1)$-resilient routing functions.}


\newcommand{\CompleteBipartiteResiliency}{For any $k$-connected complete bipartite graph there exist a set of $(k-1)$-resilient routing functions.}


\newcommand{\GeneralizedHypercubeResiliency}{For any $(i,k)$-generalized hypercube graph there exist a set of $(k^i-1)$-resilient routing functions.}


\newcommand{\ClosResiliency}{For any $k$-connected Clos network there exist a set of
$(k-1)$-resilient routing functions $(k-1)$-resilient.}


\newcommand{\GridResiliency}{For any grid graph there exist a set of $3$-resilient 
routing functions.}


\vspace{-.1in}
\subsection{Impossibility Results}\label{sect:negative}

\eat{
We first show that computing ``smart'' trees is crucial when arborescence-based routing functions are used to achieve $(k-1)$-resiliency in a $k$-connected graph. 
\begin{theorem}
\notemarco{Actually, I'm not sure this is easy to prove...} For any $k\ge6$, there exists a $k$-connected graph and a set of $k$ arc-disjoint arborescence such that any arborescence-based routing is $(k-1)$-resilient.
\end{theorem}
}
\vspace{-.05in}
 We now show that simplified forms of failover routing functions are not sufficiently powerful.
 It is well-known that without matching the incoming-edge it is not even  possible to construct $1$-resilient static routing functions~\cite{kwong-link-protection-11}. 
 To overcome this,~\cite{keep-forwarding-14} suggests to
 route packets based on a circular ordering of the edges incident at each vertex. Namely,
 a set of routing functions is  {\em \VertexCircular} if at each vertex $v$ routes a packet based on the input port and an ordered circular sequence $<e_1,\dots,e_l>$ of its incident edges as follows. If a packet $p$ is received from an edge $e_i$, then $v$  forwards it along $e_{i+1}$.  If the outgoing edge $e_{i+1}$ failed, $v$ forwards $p$ through $e_{i+2}$, and so on.
 We prove in Appendix~\ref{appe:negative-results} that this simplified routing functions cannot provably guarantee $(k-1)$-resiliency even for three-connected graphs.

\newcommand{\noCircularOrdering}{There is a $3$-connected graph $G$ for which 
no $2$-resilient \VertexCircular routing function exists.}

\begin{theorem}\label{theo:no-circular-ordering}
\noCircularOrdering 
\end{theorem}

We now exploit  the previous theorem to state another impossibility result, which shows that
the edge-connectivity between two vertices, i.e., the maximum amount of disjoint paths between the two vertices, does not match the resiliency guarantee for these two vertices. 
In other words, even if a vertex $v$ is $k$-connected to the destination (but {\em not} the entire graph), it is not possible to guarantee that a packet originated at $v$ will reach $d$ when $k-1$ edges fail. Clearly, if we want to protect against $k-1$ failures a single vertex that is $k$-connected to $d$, we can safely route along its $k$ edge-disjoint paths one after the other until the packet reaches its destination.
 However, if there are more vertices to be protected, it may be not possible to protect all of them. We say that a routing function is {\em vertex-connectivity-resilient} if each packet that is originated by a vertex $v$ that is $k$-connected to the destination $d$, can be routed towards the destination as long as less than $k$ edges fail.

 Let $G'$ be the graph obtained from $G$ by replacing each edge $e=\{x,y\}$ with $3$ edges $\{x,v_1^e\}$, $\{v_1^e,v_2^e\}$, and $\{v_2^e,y\}$, where $v_1^e$ and $v_2^e$ are new vertices added into $V(G')$. Observe that if $G$ is at least at least $2$-connected, then $G'$ is $2$-connected. Also, connectivity between the ``original'' vertices of $G$ does not change in $G'$.

\newcommand{\LemmaFromVertexCircularToVertexConnectivity}{If there exists a $2$-resilient routing function for $G'$, then there exists a \VertexCircular routing  for $G$.}

\begin{lemma}\label{lemm:from-vertex-circular-to-vertex-connectivity}
 \LemmaFromVertexCircularToVertexConnectivity
\end{lemma}

By Theorem~\ref{theo:no-circular-ordering} and Lemma~\ref{lemm:from-vertex-circular-to-vertex-connectivity}, we can easily show that vertex-connectivity-resilient is not achievable.

\newcommand{\vertexConnectivityTheorem}{There are a graph $G$ and a destination $d\in V(G)$ for which no set of vertex-connectivity-resilient routing functions exists.}

\begin{theorem}\label{theo:vertex-connectivity-impossibility}
 \vertexConnectivityTheorem
\end{theorem}

We can leverage Lemma~\ref{lemm:from-vertex-circular-to-vertex-connectivity} to show that there exists a limit on the resiliency that can be attained in a $k$-connected graph. It was proved in~\cite{podc} that perfect resiliency, i.e., resiliency against any failures that do not disconnect a sender from $d$, cannot be guaranteed. We claim a stronger bound.

\newcommand{\noKResiliency}{There is a $2$-connected graph for which no set of $2$-resilient routing functions exists.}

\begin{theorem}\label{theo:no-k-resiliency}
\noKResiliency 
\end{theorem}

Theorem~\ref{theo:no-k-resiliency} and the promising results shown in this section 
leads to the following natural and elegant conjecture that relates the $k$-connectivity of a graph to the possibility of constructing routing functions that are robust to $k-1$ edge failures.

\newcommand{\strongerConjecture}{For any $k$-connected graph, there exist a set of $(k-1)$-resilient routing functions.}

\begin{conjecture}\label{conj:stronger-conjecture}
 \strongerConjecture
\end{conjecture}

\eat{Unfortunately, despite significant effort, we have not yet resolved this conjecture, so we leave this paper on an unsatisfying, yet tantalizing, note.}

\vspace{-.1in}
\section{Randomized Routing}\label{sect:probabilistic}

\vspace{-.05in}
	In this section, we devise a set of routing functions for $G$ that is $(k - 1)$-resilient but requires a source of random bits. 
	We extend our routing function definition, which we call randomized routing (\probabilistic), as follows: a routing function 
	maps an incoming edge and the set of active edges incident at $v$ to a set of pairs $(e,q)$, where $e$ is an outgoing edge and $q$ is the probability of forwarding a packet through $e$. A packet is forwarded through a unique outgoing edge. 

	The section is structured as follows. As a prelude, we state some facts about the case when $G$ has at most $k-1$ failed edges. Then, we provide an algorithm to construct randomized routing functions, we prove it is $(k - 1)$-resilient, and show that it outperforms a simpler algorithm  in terms of expected number of next hops.

\vspace{-.1in}
\subsection{Meta-graph, Good Arcs, and Good Arborescences}\label{section:meta-graph}

\vspace{-.05in}
 The goal of this section is to provide an understanding of the structural relation between the arborescences of $\cT$ when the underlying network has at most $k - 1$ failed edges. The perspective that we build here will drive the construction of our algorithms in the following sections.
    
	We start by introducing the notion of a \emph{meta-graph}. To that end, we fix an arbitrary set of failed edges $F$. Throughout the section, we assume $|F| < k$, and define $f := |F|$.
Then, we define a meta-graph $\HF = (V_F, E_F)$ as follows:
	\begin{itemize}
		\item $V_F = \{1, \ldots, k\}$, where vertex $i$ is a representative of arborescence $T_i$.
		\item For each failed edge $e \in E$ belonging to at least one arborescences of $\cT$ we define the corresponding edge $e_F$ in $H_F$ as follows:
			\begin{itemize}
				\item $e_F := \{i, j\}$, if $e$ belongs to two different arborescences $T_i$ and $T_j$;
				\item $e_F := \{i, i\}$, i.e. $e_F$ is a self-loop, if $e$ belongs to a single arborescence $T_i$ only.
			\end{itemize}
	\end{itemize}
	Note that in our construction $\HF$ might contain parallel edges. Intuitively, the meta-graph represents a relation between arborescences of $\cT$ for a fixed set of failed edges. We provide the following lemma as the first step towards understanding the structure of $\HF$.

\newcommand{\TreeComponentsLemma}{The set of connected components of $\HF$ contains at least $k - f$ trees.}

\begin{lemma}\label{lemma:tree-components}
	\TreeComponentsLemma
\end{lemma}

    Lemma \ref{lemma:tree-components} implies that the fewer failed edges there are, the larger fraction of connected components of the meta-graph $\HF$ are trees. Note that an isolated vertex is a tree as well. In the sequel, we show that each tree-component of $\HF$ contains at least one vertex corresponding to an arborescence from which any bounce on a failed edge leads to the destination $d$ without hitting any new failed edge.

	To that end, we introduce the notion of good arcs and good arborescences. We say that an arc $(u, v)$ is a \emph{good arc} of an arborescence $T$ if on the (unique) $v$-$d$ path in $T$ there is no failed edge. Let $a = (i, j)$, for $i \neq j$, be an arc of $\dHF$, $\{u, v\}$ be the edge that corresponds to $a$, and w.l.o.g. assume $(u, v)$ is an arc of $T_j$.
	Then, we say $a$ is a \emph{well-bouncing} arc if $(u, v)$ is a good arc of $T_j$. Intuitively, a well-bouncing arc $(i, j)$ of $\dHF$ means that by bouncing from $T_i$ to $T_j$ on the failed edge $\{v, u\}$ the packet will reach $d$ via routing along $T_j$ without any further interruption. Finally, we say that an arborescence $T_i$ is a \emph{good arborescence} if every outgoing arc of vertex $i \in V_F$ is well-bouncing.
	
	\newcommand{\GoodArcsLemma}{Let $T$ be a tree-component of $\HF$ s.t.  $|V(T)| > 1$. Then, $\dT$ contains at least $|V(T)|$ well-bouncing arcs.}
	
	\begin{lemma}\label{lemma:good-arcs}
		\GoodArcsLemma
	\end{lemma}

	Now, building on Lemma \ref{lemma:good-arcs}, we prove the following.
	
	\newcommand{\AtLeastOneGoodLemma}{Let $T$ be a tree-component of $\HF$. Then, there is an arborescence $T_i$ such that $i \in V(T)$ and $T_i$ is good.}
	
	\begin{lemma}\label{lemma:at-least-one-good}
		\AtLeastOneGoodLemma
	\end{lemma}
	Let us understand what this implies. Consider an arborescence $T_i$, and a routing of a packet along it. In addition, assume that the routing hits a failed edge $e$, such that $e$ is shared with some other arborescence $T_j$. Now, if $e$ corresponds to a well-bouncing arc of $\dHF$, then by bouncing on $e$ and routing solely along $T_j$, the packet will reach $d$ without any further interruption. Lemma~\ref{lemma:at-least-one-good} claims that for each tree-component $T$ of $\HF$ there always exists an arborescence $T_i$, with $i \in V(T)$, which is good, i.e. every  failed edge of $T_i$ corresponds to a well-bouncing arc of $\dHF$.

	We can now state the main lemma of this section.
	\begin{lemma}\label{lemma:good-arborescence}
		If $G$ contains at most $k - 1$ failed edges, then $\cT$ contains at least one good arborescence.
	\end{lemma}
	\begin{proof}
		We prove that there exists an arborescence $T_i$ such that if a packet bounces on any failed edge of $T_i$ it will reach $d$ without any further interruption. 
		Let $F$ be the set of failed edges, at most $k - 1$ of them. Then, by Lemma~\ref{lemma:tree-components} we have that $\HF$ contains at least $k - f \ge 1$ tree-components. Let $T$ be one such component.
		
		By Lemma~\ref{lemma:at-least-one-good}, we have that there exists at least an arborescence $T_i$ such that 
		every outgoing arc from $i$ is well-bouncing. Therefore, bouncing on any failed arc of $T_i$ the packet will reach $d$ without any further interruption. 
	\end{proof}

\vspace{-.2in}
\subsection{An Algorithm for Randomized Routing} 

\vspace{-.05in}
Algorithm~\ref{algo:probabilistic-routing} describes our algorithm to construct a set of $(k-1)$-resilient randomized routing functions, which we call \probAlgo. The algorithm is parametrized by $q$ that we define later.
\eat{
 \vspace{.1in}
 \noindent
	\fbox{
		\begin{minipage}{0.95\columnwidth} 
        	\probAlgo: Given $\cT = \{T_1, \ldots, T_k\}$
        	\begin{enumerate}[1.]\addtolength{\itemsep}{-.4\baselineskip}
        	    \item $T :=$ an arborescence from $\cT$ sampled u.a.r.\footnote{Acronym u.a.r. stands for "uniformly at random".}
        		\item While $d$ is not reached
        		    \begin{enumerate}[1.]\addtolength{\itemsep}{-.4\baselineskip}
        			    \item Route along $T$ (canonical mode)
        				\item If a failed edge is hit then
        				    \begin{enumerate}[(a)]\addtolength{\itemsep}{-.4\baselineskip}
        					    \item With probability $q$, replace $T$ by an arborescence from $\cT$ sampled u.a.r.
        			            \item Otherwise, bounce the failed edge and update $T$ correspondingly
                            \end{enumerate}
        	        \end{enumerate}
        	\end{enumerate}
		\end{minipage}
	} 
}
%

	\vspace{.1in}
	\noindent\textbf{Correctness.}
	Assume that we, magically, know whether the arborescence we are routing along is a good one or not.
	Then, on a failed edge we could bounce if the arborescence is good, or switch to the next arborescence otherwise. And, we would not even need any randomness. However, we do not really know whether an arborescence is good or not since we do not know which edges will fail. To alleviate this lack of information we use a random guess. So, each time we hit a failed edge we take a guess that the arborescence is good, where the parameter $q$ estimates our likelihood. Notice that \probAlgo implements exactly this approach.
	As an example, consider Fig.~\ref{fig:toy-gadget}. If a packet originated at $a$ is first routed through \Orange~and the corresponding outgoing edge $e_{a,b}^F$ is failed, then the packet is forwarded with probability $q$ to an arborescence from $\cT$ sampled u.a.r. and with probability $1-q$ through  \Green, which shares the outgoing failed edge $e_{a,b}^F$ with \Red. 
	By the following theorem we show that this approach leads to $(k - 1)$-resilient routing.

	
		\newcommand{\MainTheoremProbabilistic}{Algorithm \probAlgo produces a set of $(k - 1)$-resilient routing functions.}
	
	\begin{theorem}\label{theorem:main-theorem-probabilistic}
		\MainTheoremProbabilistic
	\end{theorem}
	\begin{proof}
		By Lemma~\ref{lemma:good-arborescence} we have that there exists at least one arborescence $T_i$ of $\cT$ such that bouncing on any failed edge of $T_i$ the packet will reach $d$ without any further interruption. Now, as on a failed edge algorithm \probAlgo will switch to $T_i$ with positive probability, and on a failed edge of $T_i$ the algorithm will bounce with positive probability, we have that the algorithm will reach $d$ with positive probability.
	\end{proof}
	
\vspace{-.2in}
	\subsection{The Running Time of \probAlgo}\label{sect:running-time-probabilistic}
	
	\vspace{-.05in}
		In this subsection we analyze the expected number of times $I$ the packet is rerouted from one arborescence to another one in \probAlgo. As we are interested in providing an upper bound on $I$, 
		we make the following assumptions. First, we assume that bouncing from an arborescence which is not good the routing always bounces to an arborescence which is not good as well. Second, we assume that only by bouncing from a good arborescence the routing will reach $d$ without switching to any other arborescence. Third, we assume that there are exactly $k - f$ good arborescences, which is the lower bound provided by Lemma~\ref{lemma:tree-components} and Lemma~\ref{lemma:at-least-one-good}. Clearly, these assumptions can only lead to an increased number of iterations compared to the real case. Finally, for the sake of brevity we define $t := \tfrac{f}{k}$.
		\\
		Now, we are ready to start with the analysis. As the first step we define two random variables, where in the definitions $T$ is the arborescence variable from algorithm \probAlgo,
		\begin{eqnarray*}
				X  & := & \text{ \# of times a failed edge is hit before reaching $d$ if $T$ is not a good arborescence, and}\\
			Y  & := & \text{ \# of times a failed edge is hit before reaching $d$ if $T$ is a good arborescence}.
		\end{eqnarray*}
		
		Let $\Tinit$ be the first arborescence that we consider in \probAlgo. Then, $\EE{I}$ is upper-bounded by
 		\begin{eqnarray}
			\EE{I}  \le  \Prob{\Tinit \text{ is not good}} \EE{X}  + \Prob{\Tinit \text{ is good}} \EE{Y} , \label{eq:I}
		\end{eqnarray}
		where straight from our assumptions we have
		$$	\Prob{\Tinit \text{ is not good}} = t
		\text{, and }
			\Prob{\Tinit \text{ is good}} = 1 - t.
		$$
		Next, let us express $\EE{X}$ and $\EE{Y}$ as functions in $\EE{X}$, $\EE{Y}$, $q$, and $t$, while following our assumptions. If $T$ is not a good arborescence, then a routing along $T$ will hit a failed edge. If it hits a failed edge, with probability $1 - q$ the routing will bounce and switch to a non good arborescence. With probability $q t$ the routing scheme will set $T$ to be a non good arborescence, and with probability $q (1 - t)$ it will set $T$ to be a good arborescence. Formally, we have
		\begin{equation}\label{eq:X}
			\EE{X} = 1 + q t \EE{X} + q (1 - t) \EE{Y} + (1 - q) \EE{X}.
		\end{equation}
		Applying an analogous reasoning about $Y$, we obtain
		\begin{equation}\label{eq:Y}
			\EE{Y} = 1 + q t \EE{X} + q (1 - t) \EE{Y}.
		\end{equation}
		Observe that the equations describing $\EE{X}$ and $\EE{Y}$ differ only in the term $(1 - q) X$. This comes from the fact that bouncing on a good arborescences the packet will reach $d$ without hitting any other failed edge.
		
		By some simple calculations (see Appendix~\ref{appe:probabilistic}) we obtain:
		
\eat{		Subtracting \eqref{eq:X} from \eqref{eq:Y} we obtain
		\begin{equation}\label{eq:X-Y}
			\EE{Y} = q \EE{X}.
		\end{equation}
		Substituting \eqref{eq:X-Y} to \eqref{eq:X} gives
		\begin{equation}\label{eq:refined-X}
			\EE{X} = \frac{1}{(1 - q) q (1 - t)},
		\end{equation}
		and therefore, from \eqref{eq:X-Y},
		\begin{equation}\label{eq:refined-Y}
			\EE{Y} = \frac{1}{(1 - q) (1 - t)}.
		\end{equation}
		Substituting \eqref{eq:refined-X} and \eqref{eq:refined-Y} into \eqref{eq:I}, we obtain an upper bound on $\EE{I}$
		\begin{eqnarray}
			\EE{I} & \le & \frac{t}{(1 - q) q (1 - t)} + \frac{1}{1 - q} \label{eq:I-upper-bound}.
		\end{eqnarray}

}
		\vspace{-.2in}
		\begin{eqnarray}
			\EE{I} & \le U(q) = \frac{t}{(1 - q) q (1 - t)} + \frac{1}{1 - q}. \label{eq:I-upper-bound}
		\end{eqnarray}
\eat{
		Let $U(q)$ denote the upper-bound provided by \eqref{eq:I-upper-bound}, i.e.
		\begin{equation}\label{eq:U}
			U(q) := \frac{t}{(1 - q) q (1 - t)} + \frac{1}{1 - q}.
		\end{equation}
		Now we can prove the following lemma.

		\newcommand{\ProbabilisticRunningTimeOne}{We have that
			\[
				\EE{I} \le 2 + 4 \frac{t}{1 - t}.
			\]}

		\begin{lemma}\label{lemma:probabilistic-running-time-one}
			\ProbabilisticRunningTimeOne
		\end{lemma}
	}
		
		Note that if we know $f$ in advance, or have some guarantee in terms of an upper bound on $f$, we can derive parameter $q$ that improves the running time of  \probAlgo, as provided by the following lemma.
	
		\newcommand{\ProbabilisticRunningTimeTwo}{$U(q)$ is minimized for $q = q^* := 1 - (1 + \sqrt{t})^{-1}$, and equal to
			\begin{equation}\label{eq:opt-U}
				U(q^*) = \frac{1 + \sqrt{t}}{1 - \sqrt{t}}.
			\end{equation}}

		\begin{lemma}\label{lemma:probabilistic-running-time-two}
			\ProbabilisticRunningTimeTwo
		\end{lemma}

	Observe that $U(q^*)\le\frac{4}{1-\frac{f}{k}}$. If $f=\alpha k$, i.e., only a fraction of the edges fail, we obtain $U(q^*)\le\frac{4}{1-\alpha}$. This means that the expected number of arborescence switches does not depend on the number of failed edges but on the ratio between this number and the connectivity of the graph.
	Otherwise, if $f=k-1$, we have that the expected number of arborescence switches is bounded by $4k$, which is linear w.r.t. to the connectivity of the graph.

\newcommand{\RandomRerouting}{\textsc{Rand-Algo}\xspace}

\vspace{.1in}
\noindent\textbf{Bouncing is efficient.}
 It might be tempting to implement a variation of \probAlgo that on each failed edge switches to another arborescences chosen uar, i.e. to set $q = 1$ in the Alg.~\ref{algo:probabilistic-routing}. Let \RandomRerouting denote such a variant. The following theorem shows that  \probAlgo significantly outperforms \RandomRerouting.

\newcommand{\NeverBounce}{For any $k>0$, there exists a $2k$ edge-connected graph, a set of $2k$ arc-disjoint spanning trees, and a set of $k-1$ failed edges, such that the expected number of tree switches with \RandomRerouting is $\Omega(k^2)$. }

\begin{theorem}\label{theo:never-bounce}
 \NeverBounce
\end{theorem}

\vspace{-.2in}
\section{Packet header rewriting}\label{sect:header-rewriting}

\vspace{-.05in}
 In this section we show how to construct a set of $(k-1)$-resilient routing functions that requires only three extra bits in the packet header. We define our routing function, which we call header-rewriting (\headerrewriting) routing, as follows: in this section 
 a routing function maps an incoming edge, the set of active edges incident to $v$, and a string of bits in the header of a packet to an outgoing edge and a possibly new packet header.
 
 Consider the circular routing algorithm with the following twist. If in the circular routing the packet hits a failed edge $a$ of an arborescence $T_i$, then the packet bounces to arborescence $T_j$, if there is any, and continues routing along $T_j$. Now, if the packet hits a failed edge of $T_j$, then the packet is routed back to the edge $a$ and the circular routing continues. The corresponding algorithm is provided in Algorithm~\ref{algo:header-rewritin-routing}.

 As we show in the sequel, in case there are at most $k - 1$ failed edges then the described routing scheme delivers the packet to $d$. However, there are a few questions that we should resolve in order to implement this scheme in our routing model: first, after bouncing on a failed edge $a$ and hitting a new failed edge, how one can route the packet back to $a$; and, second, how we keep track of whether the circular routing or the one after bouncing is in use. Now, both questions could be easily answered if the packet stores the path it is routed over, which in the worst case could require ``many'' extra bits. On the other hand, as we have been discussing in the introduction, our aim is to provide a routing scheme that uses a very few bits, which we do in this section.
 
\vspace{.1in}
 \noindent\textbf{Backtracking: A routing and its inverse. }
 Essentially, the first question can be cast as a task of devising a routing scheme $R(T)$, for a given arborescence $T$, which has its inverse. Let our hypothetical scheme $R(T)$ route the packet along edges $a_1, a_2, \ldots, a_t, a_{t + 1}$ in that order. Then, the inverse routing scheme $R^{-1}(T)$ would route a packet received along $a_{t + 1}$ through edges $a_t, a_{t - 1}, \ldots, a_1$ in that order.
 %
  We choose $R(T)$ to be a DFS traversal of $T$ starting at $d$. For the sake of the traversal, we disregard the orientation of the edges of $T$, as shown in Fig.~\ref{fig:DFStraversal}.

 
 Note that
 we use canonical mode (which does not have an inverse) for routing packets along the arborescences that are chosen in the circular order. Only once the packet bounces to arborescence $T$, we 
 route the packet following scheme $R(T)$, and then follow its inverse $R^{-1}(T)$ if a new failed edge is hit, as explained above.

 \vspace{.1in}
 \noindent\textbf{Three extra bits suffices for $(k-1)$-resiliency. }
 So, to put into action our routing algorithm, we use three different routing schemes. 
 In order to distinguish which one is currently used, we store extra bits in the packet header. Those bits are used to keep the information needed to decide which routing scheme should be used.
 To keep track of which routing scheme is being used, out of the three aforementioned, we need two bits. Let $RM$ be a two-bit word with the following meaning: $RM = 0$ for canonical mode; $RM = 1$ for scheme $R(T)$; and $RM = 2$ for scheme $R^{-1}(T)$.
 
 We now motivate the usage of the third bit. Let $a$ be the last arc the packet is routed over. Then in canonical mode, i.e. if $RM = 0$, $a$ uniquely determines the arborescence along which the packet is routed. However, if $T_i$ and $T_j$, for $i < j$, share an edge $\{x, y\}$, then the arcs $(x, y)$ and $(y, x)$ are in both $R(T_i)$ and in $R(T_j)$. Therefore, if $RM \neq 0$ then the information stored in $RM$ along with $a$ is not sufficient to determine whether the arborescence the packet is routed along is $T_i$ or $T_j$. So, to keep track of whether the packet is routed along $T_i$ or $T_j$ we use another extra bit $H$. We set $H = 1$ if the packet is routed along the arborescence with higher index, i.e. along $T_j$, and set $H = 0$ otherwise.
 
 Therefore, in total, we need three additional bits (two for $RM$ and one for $H$) to keep track of which routing scheme is in use and which arborescence is currently used  to route a packet. In Appendix~\ref{appe:header-rewriting} we provide an algorithm that sets these bits precisely.

 Putting the result from Section~\ref{section:meta-graph} into the setting we have developed in this section, we show that indeed \dfalgo computes a $(k - 1)$-resilient routing.
\newcommand{\DFAlgoResiliency}{For any $k$-connected graph, \dfalgo computes a set of $(k-1)$-resilient routing functions.}
\begin{theorem}\label{theo:df-algo-resiliency}
    \DFAlgoResiliency
\end{theorem}

\vspace{-.2in}
\section{Packet duplication}\label{sect:duplication}

\vspace{-.05in}
 In this section we show that, for any $k$-connected graph $G$, it is always possible to compute duplication routing functions (\duplication) that are $(k-1)$-resilient. \duplication maps an incoming edge and the set of active edges incident at $v$ to a subset of the outgoing edges at $v$. A packet is duplicated at $v$ and one copy is sent to each of the edges in that set.

\noindent
\begin{minipage}{\textwidth} 
    \begin{minipage}{0.47\textwidth} 
        
          \begin{figure}[H]
        \centering
        \includegraphics[width=0.9\textwidth]{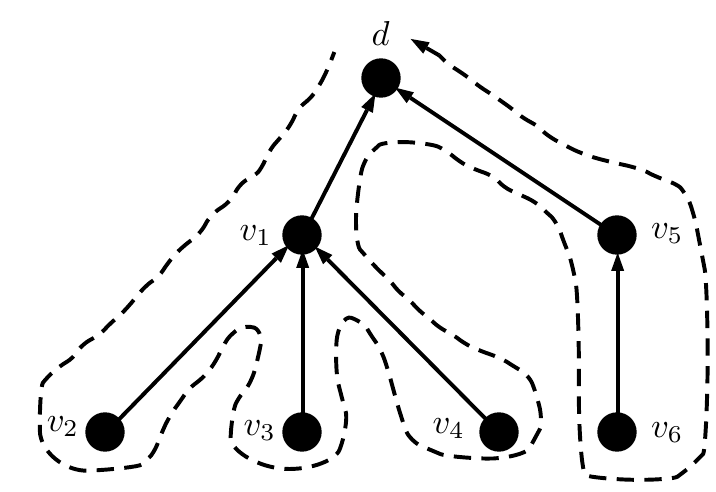}
        \caption{Let $T$ denote the arborescence on the figure. A DFS traversal is illustrated by the dashed line, i.e. $R(T) = d v_1 v_2 v_1 v_3 v_1 v_4 v_1 d v_5 v_6 v_5 d$ and $R^{-1}(T) = d v_5 v_6 v_5 d v_1 v_4 v_1 v_3 v_1 v_2 v_1 d$.}
        \label{fig:DFStraversal}
 \end{figure}   
    \end{minipage}
    \hfill
    \noindent
    \begin{minipage}{0.47\textwidth}
      \begin{algorithm}[H]
        \caption{Definition of \dupalgo.}
        \label{algo:duplication-routing}
        \begin{enumerate}
        \item $p$ is first routed along $T_1$.
        \item $p$ is routed along the same arborescence towards the destination, unless a failed edge is hit.
        \item if $p$ hits a failed edge $(x,y)$ along $T_i$, then:
        \begin{enumerate}
        \item\label{step3a} if $i < s$: one copy of $p$ is created; the original packet is forwarded along $T_{i + 1}$; the copy is forwarded along $T_l$, where $T_l$ is the arborescence that contains arc $(y,x)$.
        \item\label{step3b} if $i = s$: $s-1$ copies of $p$ are created; the original packet is forwarded along $T_{s + 1}$; the $j$'th copy, with $1\le j\le s-1$, is routed along $T_{s+j+1}$.
        \item\label{step3c} if $i > s$: $p$ is destroyed.
        \end{enumerate}
        \end{enumerate}\end{algorithm}
    \end{minipage}
\end{minipage}

\vspace{.1in}
 A naive approach would flood the whole network with copies of the same packets, i.e., each vertex $v\in V(G) $ creates a copy a packet for each outgoing edge and forwards it through that edge. There are two drawbacks to this approach. First, marking packets is necessary to avoid forwarding loops. Second, at least a copy of the packet will be routed through each edge, wasting routing resources. In the following, we present an algorithm that creates a very limited number of copies of a packet and guarantees robustness against any $k-1$ edge failures.
 
 The general idea is to carefully combine the benefits of both \CircularRouting and bounce routing (as for \headerrewriting routing in Section~\ref{sect:header-rewriting}). Circular-arborescence routing allows us to visit each arborescence, while bouncing a packet allows us to discover good arcs. (We refer the reader to Section~\ref{sect:probabilistic} for the definition of good arcs.) Bouncing packets comes at the risk of easily introducing forwarding loops as packets may be bounced between just two arborescences. Hence, we leverage our construction of  arborescences from Lemma~\ref{lemm:bipartite-trees}, which will help us to eventually hit $k-1$ distinct failed edge, and we forbid any bouncing that may create a forwarding loop.
 For simplicity, we assume that $k = 2s$ is even (see Appendix~\ref{appe:duplication} for the full proof).

Let $G$ be a $2s$-connected graph and  $T_1,\dots,T_{2s}$ be $2s$ arc-disjoint arborescences such that $T_1,\dots,T_s$ ($T_{s+1},\dots,T_{2s})$ do not share edges each other (as in Lemma~\ref{lemm:bipartite-trees}). We define the \dupalgo algorithm in Algorithm~\ref{algo:duplication-routing} and in the following show that it provides a set of $(2s - 1)$-resilient routing functions.
\eat{\begin{algorithm}
\caption{Definition of \dupalgo.}
\label{algo:duplication-routing}
\begin{enumerate}
\item $p$ is first routed along $T_1$.
\item $p$ is routed along the same arborescence towards the destination, unless a failed edge is hit.
\item if $p$ hits a failed edge $(x,y)$ along $T_i$, then:
\begin{enumerate}
\item\label{step3a} if $i < s$: two copies of $p$ are created; one copy is forwarded along $T_{i + 1}$; the other one is forwarded along $T_l$, where $T_l$ is the arborescence that contains arc $(y,x)$.
\item\label{step3b} if $i = s$: $s$ copies of $p$ are created; the $j$'th copy, with $1\le j\le s$, is routed along $T_{s+j}$.
\item\label{step3c} if $i > s$: $p$ is destroyed.
\end{enumerate}
\end{enumerate}\end{algorithm}
}
%
\newcommand{\DuplicationTheorem}{For any $2s$-connected graph and $s \ge1$, \dupalgo computes $(2s-1)$-resilient routing functions. In addition, the number of copies of a packet created by the algorithm is $f$, if $f<s$, and $2 s - 1$ otherwise, where $f$ is the number of failed edges. }
\newcommand{\DuplicationFailedArcOnEachTree}{If \dupalgo fails to deliver a packet to $d$, then each $T_i$ contains an arc that belongs to a failed edge.}
\newcommand{\DuplicationGoodLinks}{Let $T_i$ be a good arborescence from Lemma~\ref{lemma:good-arborescence}. If \dupalgo fails to deliver a packet to $d$, then $i > s$.}
\newcommand{\DuplicationThereAreKFailedLinks}{If \dupalgo fails to deliver a packet to $d$, then $T_1,\dots,T_{2s}$ contain at least $2s$ failed edges.}

We start by observing that each failed edge hit along the first $s$ arborescences cannot be a good arc, otherwise this would mean that at least a copy of a packet will reach $d$.

\begin{lemma}\label{lemm:duplication-good-links}
\DuplicationGoodLinks
\end{lemma}

By a counting argument (see Appendix~\ref{appe:duplication}), we can leverage Lemma~\ref{lemm:duplication-good-links} to prove the following crucial lemma.
\begin{lemma}\label{lemm:duplication-there-are-k-failed-links}
 \DuplicationThereAreKFailedLinks
 \end{lemma}
Lemma~\ref{lemm:duplication-there-are-k-failed-links} essentially says that if \dupalgo fails to deliver a packet to $d$, then there must be "many" failed links. That conclusion is the main ingredient in a proof (see Appendix~\ref{appe:duplication}) of the following theorem.
\begin{theorem}\label{theorem:duplication}
\DuplicationTheorem
\end{theorem}
\vspace{-.2in}
\section{Conclusions}\label{sect:conclusion}

We presented the \resilientproblem problem and explored the power of static fast failover routing in a variety of models: deterministic routing, randomized routing, routing with packet-duplication, and routing with packet-header-rewriting. We leave the reader with many interesting open questions, including resolving our conjecture that deterministic failover routing can withstand $k-1$ failures in \emph{any} $k$-connected network. Other interesting directions for future research include proving tight upper/lower bounds for the other routing models, and also considering node failures (alongside link failures).
\vspace{-.05in}

\balance

\footnotesize{
  \bibliographystyle{abbrv}
  \bibliography{alinsimple}
}

\normalsize
\appendix

\section{Deterministic Routing}\label{appe:deterministic-routing}

\subsection{$3$-connected graphs }

 \vspace{2mm}
\rephrase{Theorem}{\ref{theo:2-resiliency}}{
\TwoResiliency
}
\vspace{2mm}

\begin{proof}
 We refer to the three arc-disjoing arborescences as the $\texttt{Red}$, $\texttt{Blue}$, and $\texttt{Green}$ arborescences. 
 
 Let $<c_0,c_1,c_2>=<\texttt{Blue},\texttt{Red},\texttt{Green}>$ be an arbitrary ordering of the arborescences.
 
 \eat{\vspace{.1in}
\noindent\textbf{Technique 1. }
\begin{enumerate}
\item If $v$ originates $p$, it forwards it on an arbitrary spanning tree.
\item If $v$ receives $p$ along an edge colored $c_i$, with $i=0,1,2$, and 
the outgoing edge of the same color is active, then it 
forwards $p$ along that edge.
\item If $v$ receives $p$ along an edge colored $c_i$, with $i=0,1,2$, and 
the outgoing edge of the same color is not active, then it forwards $p$ along the 
outgoing edge with color $c_{i+1}$. If that 
edge is also not active, it forwards $p$ to the outgoing edge with color $c_{i+2}$.
\end{enumerate}
 
\begin{theorem}
 Technique 1 constructs $2$-resilient routing tables for $3$-connected graphs.
\end{theorem}
\begin{proof}
}
 We now show that a \CircularRouting routing based on this arbitrary ordering is $2$-resilient and the number of switches between trees is at most $4$.  
W.l.o.g, assume that a packet $p$ is first routed along the \Blue arborescence (see Fig.~\ref{fig:2-resiliency}).
\begin{figure}[ht]
 \begin{center}
 \includegraphics[width=.25\columnwidth]{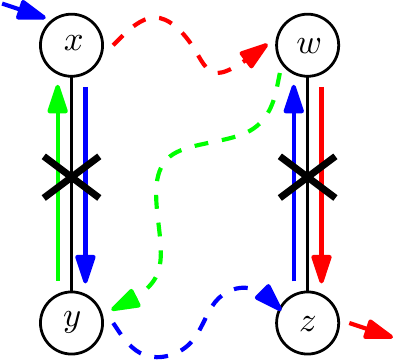}
 \caption{Proof of Theorem~\ref{theo:2-resiliency}. Dashed
 colored lines represent paths in the graph. }
 \label{fig:2-resiliency}
 \end{center}
\end{figure}
Either $p$ reaches $d$ or it hits at a vertex $x$ a failed edge $\{x,y\}$.
In the second case, $p$ is rerouted along $\texttt{Red}$. Observe that, either $p$ reaches $d$ 
or it hits at a vertex $w$, a failed edge $\{w,z\}$.
In the latter case, observe that $\{w,z\}\neq\{y,x\}$, otherwise we have a loop in the \texttt{Red}~arborescence since it contains arc $(y,x)$ and a directed path from $x$ to $y$. Observe that possibly $w=y$.
Hence, the only two failed edges are $\{x,y\}$ and $\{w,z\}$. Vertex $w$
reroutes $p$ along the \texttt{Green}~arborescence. Now, either $p$ reaches $d$ or it hits at 
a vertex $u \in \{x,y,w,z\}$, a failed edge. In the latter case,
observe that $u\neq x$, since arc $(x,y)$ belongs to the \texttt{Blue}~arborescence. Moreover, if $w\neq y$, then $u\neq w$, 
since arc $(w,z)$ belongs to the \texttt{Red}~arborescence. Moreover,
$u\neq z$, otherwise we have a loop in the \texttt{Green} arborescence since it contains arc $(z,w)$ and a directed path from $w$ to $z$. Hence, $u=y$ and $u$ reroutes $p$ along the \texttt{Blue}~arborescence. Now, observe that either $p$
reaches $d$ or it hits at $z$ the failed edge $\{w,z\}$. In the latter case, $z$ 
reroutes $p$ on \texttt{Red}. Suppose, by contradiction, 
that $p$ does not reach $d$. It means that it hits at least a failed edge along the \texttt{Red}~arborescence. However, the only arc failed along the  \texttt{Red}~arborescence is  $(w,z)$, which implies that there exists a loop in the \texttt{Red}~arborescence that contains $(w,z)$ and a directed path from $z$ to $w$---a contradiction. Hence, $p$ cannot hit any additional edge along the \texttt{Red}~arborescence, which proves the statement of the theorem.
\end{proof}

\subsection{Impossibility result for \CircularRouting routing}

 \vspace{2mm}
\rephrase{Lemma}{\ref{theo:no-circular-routing}}{
\NoCircularRouting
}
\vspace{2mm}
\begin{proof}
Consider the graph represented in Fig.~\ref{fig:counterexample-circular-routing}.
Observe that arborescences \Blue~and \Green~(\Orange~and \Red) are symmetric. Moreover, \Blue~is symmetric to \Red~and \Green~to \Orange. As a consequence, w.l.o.g., we can assume that the first arborescence where a packet originated at $c$ is routed is \Blue. Hence, there are only six different \CircularRouting routing to study: (i) $<\Blue,\Green,\Orange,\Red>$, (ii) $<\Blue,\Green,\Red,\Orange>$, (iii) $<\Blue,\Red,\Green,\Orange>$,  (iv) $<\Blue,\Red,\Orange,\Green>$, (v) $<\Blue,\Orange,\Red,\Green>$, and (vi) $<\Blue,\Orange,\Green,\Red>$.
\begin{figure}[ht]
 \begin{center}
 \includegraphics[width=.5\columnwidth]{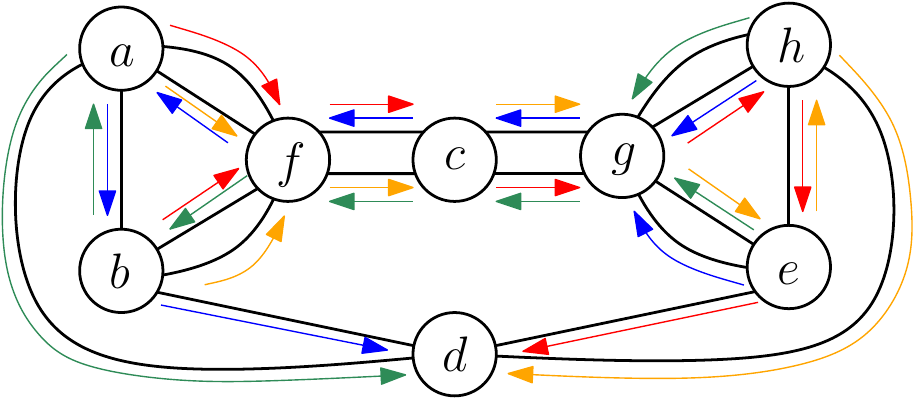}
 \caption{Counterexample used in the proof of Lemma~\ref{theo:no-circular-routing} }
 \label{fig:counterexample-circular-routing}
 \end{center}
\end{figure}
We show that in each case there exists a set of at most three edge failures such that a packet originated at vertex $c$ is forwarded along a loop. In order to distinguish between multiple edges between two vertices with adding a label $TOP$ or $DOWN$ to the edge. For instance, consider the two edges between $a$ and $f$. We refer to $\{a,f\}_{TOP}$ as the edge that contain an arc that belongs to the \Red~arborescence and to $\{a,f\}_{DOWN}$ as the edge that contain an arc that belongs to the \Blue~and \Orange~arborescences.
In case(i), if edges $\{a,f\}_{DOWN}$, $\{a,d\}$, and $\{c,g\}_{DOWN}$ fail, a packet originated at $c$ is forwarded on the following cycle $(c,f,b,a,f,c,f,b,a,\dots)$.
In case(ii), if edges $\{c,f\}_{TOP}$, $\{a,d\}$, and $\{a,f\}_{DOWN}$ fail, a packet originated at $c$ is forwarded on the following cycle $(c,f,b,a,f,b,a,\dots)$.
In case (iii), if both edges $\{a,b\}$ and $\{e,h\}$ fail, a packet originated at $c$ is forwarded on the following cycle $(c,f,a,f,c,g,h,g,c,f,b,f,c,g,e,g,c,f,a,\dots)$.
In case (iv), if  edges $\{c,f\}_{TOP}$, $\{c,g\}_{DOWN}$,  and $\{e,g\}_{TOP}$ fail, a packet originated at $c$ is forwarded on the following cycle $(c,g,c,g,\dots)$.
In case (v), if  edges $\{a,f\}_{DOWN}$, $\{c,f\}_{DOWN}$,  and $\{c,g\}_{DOWN}$ fail, a packet originated at $c$ is forwarded on the following cycle $(c,f,c,f,\dots)$.
In case (vi), if  both edges $\{a,b\}$ and $\{e,h\}$ fail, a packet originated at $c$ is forwarded on the following cycle $(c,f,a,f,c,g,e,g,c,f,b,f,c,g,h,g,c,f,a,\dots)$.
This ends the proof of the lemma for the case $k=4$. 
\end{proof}

\subsection{Constructing partially non-intersecting arborescences }
 Let $G$ be a $k$-connected graph. 
 By {\em splitting off a pair} of undirected edges $e=\{z,u\}$, $f=\{z,v\}$ we mean the operation of replacing $e$ and $f$ by a new edge connecting $u$ and $v$. 
 By {\em splitting off a vertex} $v \in V(G)$ we mean splitting off $\lceil\frac{k}{2}\rceil$ of its incident edges, removing the remainig edges, and deleting $v$ from the graph. 
 By {\em splitting off a pair of vertices} $(v,u)$, with $u,v \in V(G)$ we mean splitting off $\lfloor\frac{k}{2}\rfloor$ pair of edges incident at $u$, splitting off $\lfloor\frac{k}{2}\rfloor$ pair of edges incident at $v$, removing at least an edge connecting $u$ and $v$, and deleting both $u$ and $v$ from the graph. 
 We define the reverse operation of splitting off an edge.
 By {\em pinching an edge} $z=\{x,y\}$ to a node $v$ we mean removing $z$ from $E(G)$ and adding both $\{x,v\}$ and $\{y,v\}$ into $E(G)$.

The following lemma guarantees that we can always split off any vertex or pair of vertices in a $k$-connected graph.





\begin{lemma}\label{lemm:mader}
\cite{mader} An undirected graph $G = (V,E)$ is $k$-edge-connected if and only if $G$ can be constructed from the initial graph of two nodes connected by $k$ parallel edges by the following four operations, which keep the graph $k$-connected: 
\begin{enumerate}[(i)]
\item\label{item:op1} add an edge, 
\item\label{item:op2} pinch $\lceil\frac{k}{2}\rceil$ edges with a new node $z'$, 
\item\label{item:op3}  pinch $\lfloor\frac{k}{2}\rfloor$ edges with a new node $z'$ and add an edge connecting $z'$ with an existing node, 
\item\label{item:op4}  pinch $\lfloor\frac{k}{2}\rfloor$ edges with a new node $z'$, pinch then again in the resulting graph $\lfloor\frac{k}{2}\rfloor$ edges with another new node $z$ so that not all of these $\lfloor\frac{k}{2}\rfloor$ edges are incident to $z'$, and finally connect $z$ and $z'$ by a new edge. 
\end{enumerate}
In addition, the initial graph can be such that it contains at least an arbitrary chosen vertex of $G$.
\end{lemma}

Let $\cT = \{T_1, \ldots, T_k\}$ be a set of arborescences of $G$ rooted at $d$. Then, we say that $(T_1, \ldots, T_k)$ is a list of \emph{arc-disjoint bipartitely-edge-disjoint} (\ADBED) arborescences if the following holds:
\begin{itemize}
    \item arborescences $T_1, \ldots, T_k$ are arc-disjoint;
    \item arborescences $T_1, \ldots, T_{\lfloor\frac{k}{2}\rfloor}$ are edge-disjoint;
    \item arborescences $T_{\lfloor\frac{k}{2}\rfloor + 1}, \ldots, T_{2\lfloor\frac{k}{2}\rfloor}$ are edge-disjoint.
\end{itemize}
In other words, an \ADBED list of arborescence is a set of arc-disjoint arborescence that are in addition divided into two partitions of the equal sizes such that each of the partitions contains pairwise edge-disjoint arborescences.

We consider the case when $k$ is an even integer.

\begin{lemma}\label{lemma:ADBED-op1}
    Let $G$ be a $k$-connected graph, and $G'$ a graph obtained by applying operation~\ref{item:op1} from Lemma~\ref{lemm:mader}. If we are given a list $(T_1, \ldots, T_k)$ of \ADBED arborescences of $G$, then we can construct a list $(T_1', \ldots, T_k')$ of \ADBED arborescences for $G'$.
\end{lemma}
\begin{proof}
    The addition of an edge does not introduce any new vertex in $G$, so we set $T_i' := T_i$.
\end{proof}

\begin{lemma}\label{lemma:ADBED-op23}
    Let $G$ be a $k$-connected graph, and $G'$ a graph obtained by applying operation~\ref{item:op2} or operation~\ref{item:op3} from Lemma~\ref{lemm:mader}. If we are given a list $(T_1, \ldots, T_k)$ of \ADBED arborescences of $G$, then we can construct a list $(T_1', \ldots, T_k')$ of \ADBED arborescences for $G'$.
\end{lemma}
\begin{proof}
Let $z'$ denote the vertex added to $G$ in order to obtain $G'$. Initially, we let $T_i' := T_i$, and then modify each $T_i'$, so that $(T_1', \ldots, T_k')$ is a list of \ADBED arborescences of $G'$, in two phases. In the first phase we alter each $T_i'$ that contains a pinched arc, and in the second phase we modify the remaining ones.

\textbf{The first phase.} For each edge $e=\{x,y\} \in (E(G) \setminus E(G'))$, i.e. for each pinched edge, if arc $(x,y)$ belongs to $T_i$, let $e_1=\{x,z'\}$ and $e_2=\{y,z'\}$ be the two edges that are split off from $G'$ in order to obtain $e$. We then add arcs $(x,z')$ and $(z',y)$ to $T_i'$ and remove $(x, y)$.
\\
If after the changes any $T_i'$ is not an arborescence, we remove outgoing edges at $z'$ until $T_i'$ is an arborescence. This can be done by simply breaking cycles at $z'$ and removing multiple paths from $z'$ to $d$ at $z'$.




Now, we show some properties of the currently obtained $T_1', \ldots, T_k'$.

First, observe that $z'$ has at most one outgoing arc in each of the arborescences as we remove all the cycles, and parallel paths from $z'$ to $d$.

Second, by the construction of $T_1', \ldots, T_k'$ and the properties of $T_1, \ldots, T_k$ we have that each edge incident to $z'$ is shared by at most one arborescence in $\{T_1',\dots,T_{\lfloor\frac{k}{2}\rfloor}'\}$ and at most one arborescence in $\{T_{\lfloor\frac{k}{2}\rfloor+1}',\dots,T_{2\lfloor\frac{k}{2}\rfloor}'\}$.

Third, observe that there are at most $k/2$ incoming arc at $z'$ belonging to $T_1',\dots,T_{\lfloor\frac{k}{2}\rfloor}'$ ($T_{\lfloor\frac{k}{2}\rfloor+1}',\ldots,T_{2\lfloor\frac{k}{2}\rfloor}'$). If it would not be the case,
then it would mean that at least an edge in $E(G') \setminus E(G)$ is shared by two arborescences among $T_1',\ldots,T_{\lfloor\frac{k}{2}\rfloor}'$ ($T_{\lfloor\frac{k}{2}\rfloor + 1}',\ldots,T_{2\lfloor\frac{k}{2}\rfloor}'$) in $G$. However, that would contradict, along with out construction of $T_1', \ldots, T_k'$ would contradict that $(T_1, \ldots, T_k)$ is \ADBED of $G$.

\textbf{The second phase.} For each arborescence $T_i'$ that has no outgoing arcs at $z'$, we do as follows. W.l.o.g., assume that $i \le \lfloor k/2\rfloor $. We add into $T$ an arbitrary outgoing arc at $z'$ such that the symmetric incoming arc is not contained in any tree in $\{T_1', \ldots, T_{\lfloor\frac{k}{2}\rfloor}'\}$.

Next, our goal is to argue that there always exists an edge $\{x,y\} \in N(z')$ that is not shared by any arborescence in $\{T_1',\ldots,T_{\lfloor\frac{k}{2}\rfloor}'\}$.
\\
 Observe that $z'$ has at least $k$ incident edges. 
 On the other hand, as we have noted, there exist at most $\lfloor k/2\rfloor$ incoming arcs at $z'$ belonging to $T_1', \dots, T_{\lfloor\frac{k}{2}\rfloor}'$ and, at most $\lfloor\frac{k}{2}\rfloor-1$ outgoing arcs that belong to $T_1',\ldots,T_{\lfloor\frac{k}{2}\rfloor}'$. Hence, there exist at least $k-(\lfloor k/2\rfloor + \lfloor k/2\rfloor -1) \ge 1$ edges that are not shared by any of $T_1',\ldots,T_{\lfloor\frac{k}{2}\rfloor}'$. This means that there exists an arc $(x,z')$ that is not shared by any arborescence among $T_1',\ldots,T_{\lfloor\frac{k}{2}\rfloor}'$.

This completes the proof.
\end{proof}

\begin{lemma}\label{lemma:ADBED-op4}
    Let $G$ be a $k$-connected graph, and $G'$ a graph obtained by applying operation~\ref{item:op4} from Lemma~\ref{lemm:mader}. If we are given a list $(T_1, \ldots, T_k)$ of \ADBED arborescences of $G$, then we can construct a list $(T_1', \ldots, T_k')$ of \ADBED arborescences for $G'$.
\end{lemma}
\begin{proof}
    In this case, we have two additional vertices $z'$ and $z$. After we pinch at least $\lfloor\frac{k}{2}\rfloor$ edges to $z'$, we do the same modifications applied for operations \ref{item:op2} and \ref{item:op3}. Since the degree of $z'$ may be $k-1$ at most one arborescence $T_h'$ will not have an outgoing arc at $z'$. After we pinch at least $\lfloor\frac{k}{2}\rfloor$ edges to $z$, we do the same modifications applied for operations \ref{item:op2} and \ref{item:op3}. Since the degree of $z$ may be $k-1$ at most one arborescence $T_j'$ will not have an outgoing arc at $z$.
    
After that, we add an edge between $z$ and $z'$. If $j \neq h$, we can safely add arc $(z,z')$ into $T_j'$ and arc $(z',z)$ into $T_h'$. If $j=h$, we cannot add both arcs $(z,z')$ and $(z',z)$ into $T_j'$, because it induces a cycle. W.l.o.g., let assume that $1 \le j \le \lfloor\frac{k}{2}\rfloor$ or $j=k$.
\\
We therefore consider an arbitrary arborescence $T_f'$, where $1\le f\neq j \le \lfloor\frac{k}{2}\rfloor$. We add either $(z,z')$ or $(z',z)$ into $T_f'$ in such a way that $T_f'$ is a directed acyclic graph. This can always be done. W.l.o.g, let $(z,z')$ be the arc added into $T_f'$. We then remove the outgoing arc $(z',x)$ of $T_f'$ from $T_f'$ and add it into $T_j'$. We also add $(z',z)$ into $T_j'$.

This completes the construction
\end{proof}

\vspace{2mm}
\rephrase{Lemma}{\ref{lemm:bipartite-trees}}{
\SmartTreesTheorem
}
\vspace{2mm}
\begin{proof}
    We prove the lemma by the induction on the number of applied operations.
    


 \textbf{The base case $i = 0$.} Graph $G_0$ contains the destination vertex $d$ and another vertex $v$. Since $G_0$ is $k$-connected, these two vertices  are connected by at least $k$ parallel edges $e_1, \dots, e_{k}$. For each $j = 1, \ldots, k$, we assign $e_{j}$ to $T_{j}$ and orient each arc towards $d$. Hence, the lemma trivially holds for $G_0$.




 \textbf{The inductive step $i \ge 1$.} Let $G_{i-1}$ be a $k$-connected graph and let $\cT_i$ be its \ADBED list of arborescences.

 
 Let $G_{i}$ be a graph obtained from $G_{i-1}$ applying any of the four operations described in Lemma~\ref{lemm:mader}. Then, from Lemma~\ref{lemma:ADBED-op1}, Lemma~\ref{lemma:ADBED-op23} and Lemma~\ref{lemma:ADBED-op4} it follows that we can construct an \ADBED list of arborescences for $G_i$ as well.
 
 Hence, the statement of our main lemma holds.
\end{proof}

\eat{\vspace{2mm}
\rephrase{Lemma}{\ref{lemm:bipartite-trees}}{
\SmartTreesTheorem
}
\vspace{2mm}}

\subsection{$4$-connected graphs }

\vspace{2mm}
\rephrase{Theorem}{\ref{theo:3-resiliency}}{
\ThreeResiliency
}
\vspace{2mm}

\begin{proof}
A packet $p$ is routed along $T_1$ (see Fig.~\ref{fig:3-resiliency}).
\begin{figure}[ht]
 \begin{center}
 \includegraphics[width=.5\columnwidth]{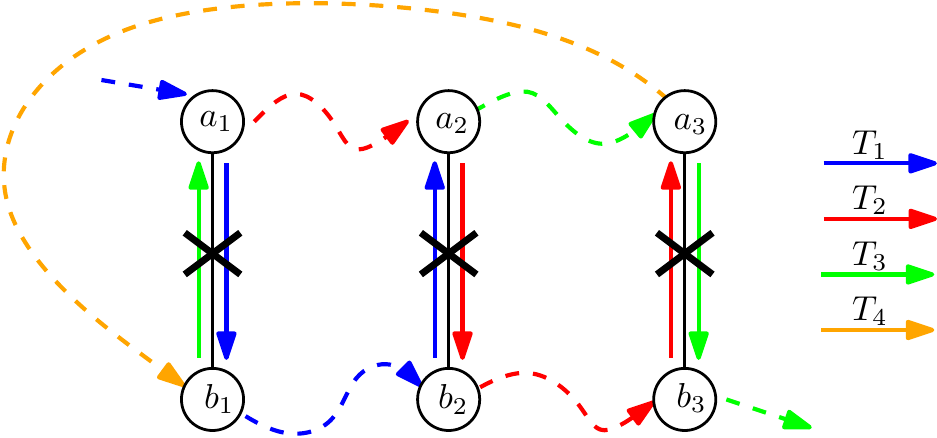}
 \caption{Proof of Theorem~\ref{theo:3-resiliency}. Dashed colored lines represent paths in the graph. }
 \label{fig:3-resiliency}
 \end{center}
\end{figure}
It either reaches the destination vertex $d$ or it hits a failed edge $e_1=\{a_1,b_1\}$ at $a_1$. In the latter case, it is rerouted along $T_2$. It either reaches $d$ or it hits a failed edge $e_2=\{a_2,b_2\}$ at $a_2$. In the latter case, observe that $e_1$ is a distinct edge from $e_2$, otherwise if $\{a_1,b_1\}=\{b_2,a_2\}$, we have a cycle in $T_2$. Hence, $p$ is routed along $T_3$. It either reaches $d$ or it hits a failed edge $e_3=\{a_3,b_3\}$ at $a_3$. In the latter case, observe that $e_3$ is a distinct edge from both $e_1$ and $e_2$, otherwise if $\{a_2,b_2\}=\{b_3,a_3\}$, we have a cycle in $T_3$ and if $\{a_1,b_1\}=\{b_3,a_3\}$ then $T_3$ shares an edge with $T_1$---a contradiction. Hence, $p$ is routed along $T_4$. It either reaches $d$ or it hits a failed edge $e^* \in \{\{b_1,a_1\},\{b_2,a_2\},\{b_3,a_3\}\}$. If $e^*=\{b_3,a_3\}$, $T_4$ contains a cycle---a contradiction. If $e^*=\{b_2,a_2\}$, $T_4$ shares an edge with $T_2$---a contradiction. Hence, $e^*=\{b_1,a_1\}$ and $p$ is rerouted along $T_1$. It either reaches $d$ or it hits a failed edge $e' \in \{\{a_1,b_1\},\{b_2,a_2\},\{b_3,a_3\}\}$. If $e'=\{a_1,b_1\}$, $T_1$ contains a cycle---a contradiction. If $e'=\{b_3,a_3\}$, $T_1$ shares an edge with $T_3$---a contradiction. Hence, $e'=\{b_2,a_2\}$ and $p$ is rerouted along $T_2$. It either reaches $d$ or it hits a failed edge $\bar e \in \{\{a_2,a_2\},\{b_3,a_3\}\}$. If $\bar e=\{a_2,b_2\}$, $T_2$ contains a cycle---a contradiction. Hence, $\bar e=\{b_3,a_3\}$ and $p$ is rerouted along $T_3$. It either reaches $d$ or it hits the failed edge $\{a_3,b_3\}$, which is not possible since $T_3$ does not contain a cycle. Hence $p$ reaches $d$.
\end{proof}

\subsection{$5$-connected graphs }

\vspace{2mm}
\rephrase{Theorem}{\ref{lemm:plus-one-resiliency}}{
\FourResiliencyLemma
}
\vspace{2mm}

\begin{proof}
We prove that $R$ is $c$-resilient. First we route a packet $p$ along $T_k$. If $p$ hits a failed edge $\{x,y\}$ at $x$, we switch to \CircularRouting routing based on arborescences $T_1,\dots,T_{k-1}$ starting from the arborescence that contains arc $(y,x)$. 
 Suppose, by contradiction, that routing is not $c$-resilient, i.e., a forwarding loop arises with less than $c+1$ link failures. Let $e_i=(a_i,b_i)$, with $i=1,\dots,r-1$, be the $i$'th failed arc hit by a packet $p$. Let $T_i$ be the arborescence that contains arc $(b_1,a_1)$. 
 Two cases are possible: (i) the forwarding loop hits edge $\{a_1,b_1\}$ or (ii) not. In case (i), consider the scenario in which only edges $\{a_2,b_2\},\dots,\{a_c,b_c\}$ failed. If a packet $p$ is originated by $a_1$ and it is initially routed along $T_i$, if it hits $(b_1,a_1)$, since this arc is not failed, $p$ will enter a forwarding loop, which is a contradiction since we assumed that the \CircularRouting routing is $(c-1)$-resilient. Hence, the forwarding loop does not hit arc $(a_1,b_1)$. Analogously, in case (ii), consider the scenario in which only edges $\{a_2,b_2\},\dots,\{a_c,b_c\}$ failed. Since the forwarding loop does not hit  $(a_1,b_1)$, we have a contradiction since we assumed that the \CircularRouting routing is $(c-1)$-resilient. Hence, our routing scheme is $c$-resilient.
\end{proof}

\section{Constrained topologies}\label{appe:specific-networks}

\subsection{\nameproptitle routing functions}
 In this section we introduce a trivial sufficient condition to achieve $(k-1)$-resiliency and show that several well-known topologies admits arborescences that satisfy it.
 
%

\begin{definition}
 A routing function is $k$-\nameprop if there exist $k$ arc-disjoint arborescences $T_1,\dots,T_k$ such that if a packet $p$ is routed along an arborescence $T_i$, with $i=1,\dots,k$, and hits a failed edge, then $p$ hit at least $i$ distinct failed edges. In addition, a packet either reaches $d$ or hits a failed outgoing edge in the $k$'th arborescence.
\end{definition}

 
 
Observe that a  $k$-\nameprop routing function might sometimes reroute a packet to another arborescence even if that packet did not hit a failed edge. This is crucial, for instance, in our construction of a set of $(k-1)$-resilient routing functions for generalized hypercubes below. Hence, routing is not arborescence-based as in previous sections. This will be necessary only for hypercube topologies.
 
 \begin{theorem}
  A $k$-\nameprop routing function $f$ is $(k-1)$-resilient.
 \end{theorem}
 
\begin{proof}
 Suppose, by contradiction, that a packet $p$ is trapped in a forwarding loop. Since $f$ is $k$-\nameprop, a packet $p$ will eventually either reach $d$ or hit a failed edge along the $k$'th arborescence 
 which means that $k$ edges failed---a contradiction.
\end{proof}
 
 We now give a high-level description of how to construct $k$-\nameprop routing functions
 for several topologies: cliques, complete bipartite graphs, and chordal graphs~\cite{diestel},
 generalized hypercubes, Clos networks~\cite{fattree}.
 Each of these graphs do not contains multiple edges.
Our high-level technique consists of two steps. 
 First, the input graph is recursively decomposed into smaller substructures (possibly  with less connectivity than the original graph) 
 for which it is easier to compute a $k$-\nameprop routing function. Then,
 these substructures are interconnected in such a way that the 
 resiliency is retained or even increased. The main challenge is to retain the  $k$-\nameprop of the routing functions during the interconnection phase.

\subsection{Clique graphs}\label{sect:clique}
 A \emph{clique} of size $k$ consists of $k$ vertices all connected to each other. Since there exists $k-1$ edge-disjoint paths between every two pair of vertices, a clique of size $k$ is $k-1$ connected.


\begin{theorem}\label{theo:clique-resiliency}
\CliqueResiliency
\end{theorem}

\begin{proof}
 Let $v_1,\dots,v_k,d$ be the set of $k+1$ vertices, where $d$ is the destination vertex. 
 We construct a $k$-\nameprop routing function based on $k$ arc-disjoint arborescences  $T_1,\dots,T_k$ as follows. For each $i=1,\dots,k$, add $(v_i,d)$, $(v_1,v_i),\dots,(v_{i-1},v_i),(v_{i+1},v_i),(v_k,v_i)$ into $T_i$. 
 Routing is as follows. A packet is first routed along $T_1$. A packet is routed along $T_i$, with $i=1,\dots,k$ as long as it does not hit a failed edge. In that case, $p$ is rerouted along $T_{i+1}$.
 
 Suppose, by contradiction, that this is not a $k$-\nameprop routing function, i.e., either (i) a packet $p$ is routed along an arborescence $T_i$, with $i=1,\dots,k$, and hits a failed edge, but $p$ hit only $i-1$ distinct failed edges or (ii) a packet does not reach $d$ and does not hit a failed edge in the $k$'th arborescence $T_k$. 
 
 Case (ii) is not possible since a packet is rerouted on a different arborescence every time a failed edge is hit.
 
 In case (i), let $e=\{v_i,v_j\}$ be the first failed edge that is hit by $p$ twice. Clearly, $p$ cannot hit $e$ twice in the same direction, otherwise it means that $p$ has been rerouted $k$ times without hitting a failed edge twice---a contradiction. Hence, $p$ hits $e$ in two opposite directions, i.e. from $v_i$ to $v_j$ and from $v_j$ to $v_i$.  W.l.o.g., let $i<j$. Before $p$ reaches $v_j$ we have that it hit $j-1$ distinct failed edges. In addition, since $v_j$ routes $p$  to $v_i$ along $T_i$, we have that all its edges to $d,v_{j+1},\dots,v_{k}$ failed. All these $(j-1)+(k-j+1)=k$ failed edges are distinct, otherwise $e$ is not the first edge that $p$ hits twice---a contradiction. Hence, the statement of the theorem holds in case (ii) as well. 
\end{proof}

 We prove another property that will be used later in this section. Let $n_i$, with $1,\dots,k$, be the only neighbor of $d$ such that $(n_i,d)$ belongs to $T_i$. 

\begin{lemma}\label{lemm:neighbor-ordering-lemma-clique}
For any $k$-connected clique graph there exists a set of $(k-1)$-resilient routing functions such that if a packet is routed at a vertex $n_i$ along $T_i$, then it does not traverse any vertex $n_1,\dots,n_{i-1}$ while it is routed through $T_i,\dots,T_k$.
\end{lemma}

\begin{proof}
 Consider the same routing solution used in the proof of Theorem~\ref{theo:clique-resiliency}. Each vertex $n_i$ is a leaf of each arborescence $T_j$, with $i\neq j$. Hence, a packet is never routed towards $n_i$, unless a packet is routed along $T_i$. Since a packet is rerouted only from an arborescence $T_i$ to an arborescence $T_{i+1}$, we have the statement of the theorem.
\end{proof}

\subsection{Complete Bipartite Graphs}\label{sect:complete-bipartite}
A \emph{complete bipartite} graph $G=(A,B,E)$ consists of $|A|+|B|$ vertices $a_1,\dots,a_{|A|}$, $b_1,\dots,b_{|B|}$ and there exists an edge between every pair of vertices $a_i$ and $b_j$, with $i=1,\dots,|A|$ and $j=1,\dots,|B|$. A $(A,B,E)$ complete graph is $k$-connected, where $k=\min\{|A|,|B|\}$.

We prove the following theorem.


\begin{theorem}\label{theo:complete-bipartite-resiliency}
\CompleteBipartiteResiliency
\end{theorem}

\begin{proof}
 We construct a $k$-\nameprop routing function. W.l.o.g., assume that $d$ is in $A$. Let $k=\min\{|A|,|B|\}$. For each $i=1,\dots,k$, add into $T_i$ arcs $(b_i,d),(a_1,b_i),\dots,(a_{|A|},b_i),(b_1,a_i),$ $\dots,(b_{i-1},a_i),(b_{i+1},a_i),\dots,(b_{|B|},a_i)$. Routing is performed exactly as for cliques (refer to the proof of Theorem~\ref{theo:clique-resiliency}). We now prove that this is a $k$-\nameprop routing function. Suppose, by contradiction, that this is not a $k$-\nameprop routing function, i.e., either (i) a packet $p$ is routed along an arborescence $T_i$, with $i=1,\dots,k$, and hits a failed edge, but $p$ hit only $i-1$ distinct failed edges or (ii) a packet does not reach $d$ and does not hit a failed edge in the $k$'th arborescence $T_k$. Clearly, case (ii) is not possible, as in the proof of Theorem~\ref{theo:clique-resiliency}. 
 In case (i), let $T_i$ be the arborescence along which a packet $p$ first hits a failed edge $e$ twice. Clearly, $p$ already hit $i-1$ distinct failed edges. Moreover, it hits $e$ in two opposite directions, otherwise, $p$ would be rerouted from an arborescence $T_j$ to an arborescence $T_i$, with $j<i$, which means that at least $k$ distinct edges failed---a contradiction. Hence, we have two cases, either (a) $e=(b_i,d)$ failed or (b) $e=(a_j,b_i)$ failed, with $1\le i\le |B|$ and $1\le j\le |A|$. Since $p$ hits $e$ in two opposite directions, case (a) is not possible. In case (b), observe that a packet is routed to $a_j$ ($b_i$) only if it is routed along $T_j$ ($T_i$). Hence, since $p$ cannot be routed from $a_j$ to $b_i$ ($b_i$ to $a_j$) and $b_i$ ($a_j$) is a leaf of every arborescence $T_l$, with $l\neq i$ ($l\neq j$), a packet will be rerouted to $b_i$ ($a_j$) only after it is rerouted along the other arborescences, which implies that there are at least $k$ distinct failed edges---a contradiction. 
\end{proof}

\subsection{Generalized hypercubes} A generalized hypercube is defined recursively as follows. A clique of size $k+1$ is a $(1,k)$-\emph{generalized hypercube}. A $(i,k)$-generalized hypercube, with $i>1$, consists of $k+1$ copies of a  $(i,k)$-generalized hypercube where all copies of the same vertex form a clique 
 (of size $k+1$). 
 Observe that the connectivity of a generalized hypercube increases by a factor of $k$
 at each recursive step. Hence, a $(i,k)$-generalized hypercube is a $k^i$-connected
 graph. 
 The construction of a set of $k^i$-\nameprop routing functions 
 is done recursively. 
 First, we construct a $k$-\nameprop routing for a clique of size $k+1$. 
 Then, in the recursive step, we interconnect all the smaller copies and combine the existing routing functions in such a way that  
 the resiliency of the graph is increased by a factor of $k$ while retaing the $k^i$-\nameprop property.


\begin{theorem}\label{theo:generalized-hypercube-resiliency}
\GeneralizedHypercubeResiliency
\end{theorem}

\begin{proof}
We denote by $H(i,k,l)$ a graph containing $l$ copies of a $(i,k)$-generalized hypercube where all copies of the same vertex form a clique. Observe that $H(i,k,k+1)=H(i+1,k)$. 
A $H(i,k,l)$ graph is $(k^i+l-1)$-connected.
We recall that we denote by $n_j$, with $1\le j\le k^i+l-1$, a neighbor of $d$ such that $(n_j,d)$ belongs to the $j$'th arc-disjoint arborescence $T_j$.
We prove that there exists a set of $(k^i+l-1)$-\nameprop routing functions for $H(i,k,l)$ by induction on $i$ and $l$. Moreover, we also prove that if a packet is routed at a vertex $n_i$ along $T_i$, then it does not traverse any vertex $n_1,\dots,n_{i-1}$ while it is routed through $T_i,\dots,T_k$.

In the base case, $H(1,k,1)$ is a clique of size $k+1$. W.l.o.g., by symmetry of the hypercube construction, we assume that the destination vertex $d$ is contained in this clique. By Theorem~\ref{theo:clique-resiliency}, there exists a $k$-\nameprop routing function. Moreover, by Lemma~\ref{lemm:neighbor-ordering-lemma-clique}, we have that if a packet is routed at a vertex $n_i$ along $T_i$, then it does not traverse any vertex $n_1,\dots,n_{i-1}$ while it is routed through $T_i,\dots,T_k$.

In the inductive step, we construct a $c+1$-\nameprop routing function for $H^{l+1}=H(i,k,l+1)$, where $1\le l\le k$ and $c=(k^i+l-1)$.  Graph $H^{l+1}$ consists of a graph $H^l=H(i,k,l)$ and a graph $H^1=H(i,k,1)$, where each vertex of $H^1$ is connected to $l$ vertices of $H^l$ and each vertex of $H^l$ is connected to exactly one vertex of $H^1$. Destination $d$ is in $V(H^l)$ and we denote by $d^1$ the only neighbor of $d$ in $V(H^1)$.
 By inductive hypothesis, there exists a $c$-\nameprop routing function for $H^l$, which is $c$-connected, based on $c$ arc-disjoint arborescences $T_1^l,\dots,T_c^l$. By inductive hypothesis, there exists a set of $k^i$-\nameprop routing functions for $H^1$, which is $k^i$-connected, based on $k^i$ arc-disjoint arborescences $ T_1^1,\dots, T_{k^i}^1$. 

 We now construct $c+1$ arc-disjoint arborescences $T_1,\dots,T_{c+1}$ of $\vec H^{l+1}$.
 For each $j=1,\dots,k^i$, let $n_j^1$ be the unique neighbor of $d^1$ such that arc $(n_j^1,d^1)$ belongs to $T_j^1$. Let $N^1$ be the set $\{n_1^1,\dots,n_{k^i}^1\}$.
 Let $n_{j,1},\dots,n_{j,l-1}$ be the neighbors of $n_j^1$ that belong to $H^l$ and $v_1,\dots,v_{l-1}$ be the neighbors of a vertex $v \in V(H^1)$ that belong to $H^l$.
 For each $j=1,\dots,k^i-1$, let $T_j$ be the union of both $T_j^1$ and $T_j^l$  plus arc $(n_j^1,n_{j,1})$.
 Moreover, for each $j=1,\dots,k^i-1$, we reverse arc $(n_j^1,d^1)$ in $T_j$.
 For each $j=k^i,\dots,k^i+l-2$, let $T_{j}$ be the union of both $ T_j^l$ and, for each vertex $v\in V(H^1)$, arc  $(v,v_{j-k^i+2})$, which connects $H^1$ to $H^l$. 
 Arborescence $T_{k^i+l}$ is constructed from a copy of $T_{k^i}^1$ 
 by adding all the edges from vertices in $V(H^l)$ to vertices in $V(H^1)$ and reversing arc $(d,d^1)$. 
 Arborescence $T_{k^i+l-1}$ is constructed from a copy of $T_{k^i+l}^l$ 
 by adding all the arcs from vertices in $V(H^1)$ to vertices in $V(H^l)$ that do not belong to any $T_1,\dots,T_{k^i+l-2},T_{k^i+l}$.
 Moreover, we add into $T_{k^i+l-2}$ arc $(d^1,n_{k^i}^1)$ and all arcs $( n_j^1,d)$, with $j=1,\dots,k^i-2,k^i$.
 
 Routing is as follows. Routing at vertices of $H^l$ is unchanged, i.e., if a vertex was routing from an arborescence $T_j^l$ towards an arborescence $T_{j+h}^l$, now it routes a packet received through $T_j$ along $T_{j+h}$. In addition, if a packet cannot be routed along $T_{k^i+l-1}$, then it is rerouted through $T_{k^i+l}$ and if a packet is received from $H^1$, it is rerouted through $T_1$, unless $p$ is received from $T_{k^i+l}$. Routing at vertices of $H^1$ is unchanged, i.e., if a vertex was routing from an arborescence $T_j^1$, with $j=1,\dots,k^i-1$,  towards an arborescence $T_{j+h}^1$, now it routes a packet received through $T_j$ along $T_{j+h}$ and if a packet cannot be routed along $T_{j}$, then it is rerouted through the next available arborescence. 

 Consider the example in Figure~\ref{fig:hypercube-partitioning-base-case}, where the base case for a  $(i,1)$-generalized hypercube is depicted together with its unique arc-disjoint arborescence $T_1$. In order to construct a $2$-\nameprop routing function for the $(2,1)$-generalized hypercube depicted in Figure~\ref{fig:hypercube-partitioning-2-case}, we create a copy of a $(1,1)$-generalized hypercube $H^l=H(1,1,1)$, denoted by $H^1$, and construct $T_2$ using $T_1^1$. We add all arcs from $H^l$ to $H^1$ in $T_2$, but we reverse $(d,d^1)$. We then add $(n_1^1,n_1)$ and $(d^1,n_1^1)$ into $T_1$. When a packet is routed from $H^l$ to $H^1$, it is rerouted along $T_1$, which is obvious for edge $(n_1^1,n_1)$, unless it is routed from $d^1$ to $d$.
 \begin{figure}[tb]
  \centering
    \includegraphics[width=.25\columnwidth]{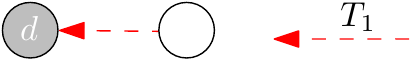}
    \caption{A $(1,1)$-generalized hypercube with one arc-disjoint arborescence (solid black ). }
\label{fig:hypercube-partitioning-base-case}
\end{figure}
\begin{figure}[tb]
  \centering
    \includegraphics[width=.25\columnwidth]{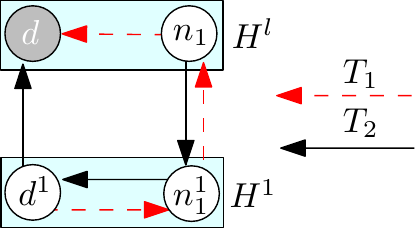}
    \caption{A $(2,1)$-generalized hypercube with two arc-disjoint arborescences (solid black and dashed red). }
\label{fig:hypercube-partitioning-2-case}
\end{figure}
We now construct a $3$-\nameprop routing function  for the $(3,1)$-generalized hypercube depicted in Figure~\ref{fig:hypercube-partitioning-induction}. We create a copy of a $(2,1)$-generalized hypercube $H^l=H(1,1,2)=H(2,1,1)$, denoted by $H^1$. We construct $T_1$ by interconnecting $T_1^l$ and $T_1^1$ with an edge $(n_1^1,n_1)$ and reversing $(n_1^1,d^1)$. We construct $T_3$ from $T_2^1$. We add all edges from $H^l$ to $H^1$ in $T_3$, but we reverse $(d,d^1)$. We then construct $T_2$ by interconnecting $T_2^l$ with edges from $H^1$ to $H^l$, except when the edge is incident to a vertex in $\{d^1,n_1^1\}$. We add edges $(n_1^1,d^1)$ and $(d^1,n_2^1)$.
When a packet is routed from $H^l$ to $H^1$, it is rerouted along $T_1$, unless it is routed from $d^1$ to $d$. For instance when a packet is forwarded from $n_2^1$ to $n_2$ along $T_2$, it is then forwarded along $T_1$. 
\begin{figure}[tb]
  \centering
    \includegraphics[width=.45\columnwidth]{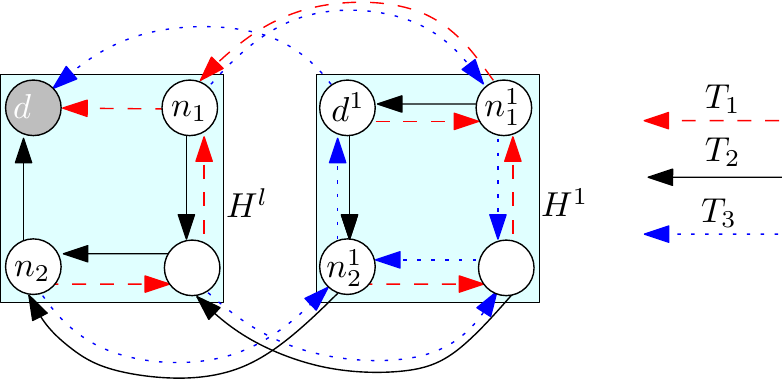}
    \caption{A $(3.1)$-generalized hypercube with three arc-disjoint arborescences (solid black, dashed red, and dotted blue). }
\label{fig:hypercube-partitioning-induction}
\end{figure}
We now construct a $4$-\nameprop routing function for the $(4,1)$-generalized hypercube depicted in Figure~\ref{fig:hypercube-partitioning-induction-16}. We create a copy of a $(3,1)$-generalized hypercube $H^l$, denoted by $H^1$. We construct $T_i$, with $i=1,2$, by interconnecting $T_i^l$ and $T_i^1$ with an edge $(n_i^1,n_i)$ and reversing $(n_i^1,d^1)$.  We construct $T_4$ using $T_3^1$. We add all edges from $H^l$ to $H^1$ in $T_4$, but we reverse $(d,d^1)$. We then construct $T_3$ by interconnecting $T_3^l$ with edges from $H^1$ to $H^l$, except when the edge is incident to a vertex in $\{d^1,n_1^1,n_2^1,\}$. We add edges $(n_1^1,d^1)$,$(n_2^1,d^1)$, and $(d^1,n_3^1)$.
 When a packet is routed from $H^l$ to $H^1$, it is rerouted along $T_1$, unless it is routed from $d^1$ to $d$. Observe that a packet that is sent along $T_2$ from $x_{4,3}$ to $n_2^1$, where $x_{i,j}$ is the vertex on $i$'th row and $j$'th column in Figure~\ref{fig:hypercube-partitioning-induction-16}, it is rerouted on $T_1$ because vertex $n_2$ reroutes packet received from $x_{2,3}$ along $T_1^l$. On the contrary, since vertex $d$ does not reroute a packet received from $n_3$ along $T_1^l$, also vertex $d^1$ does not reroute a packet received from $n_3^1$ along $T_1$.
\begin{figure}[tb]
  \centering
    \includegraphics[width=.45\columnwidth]{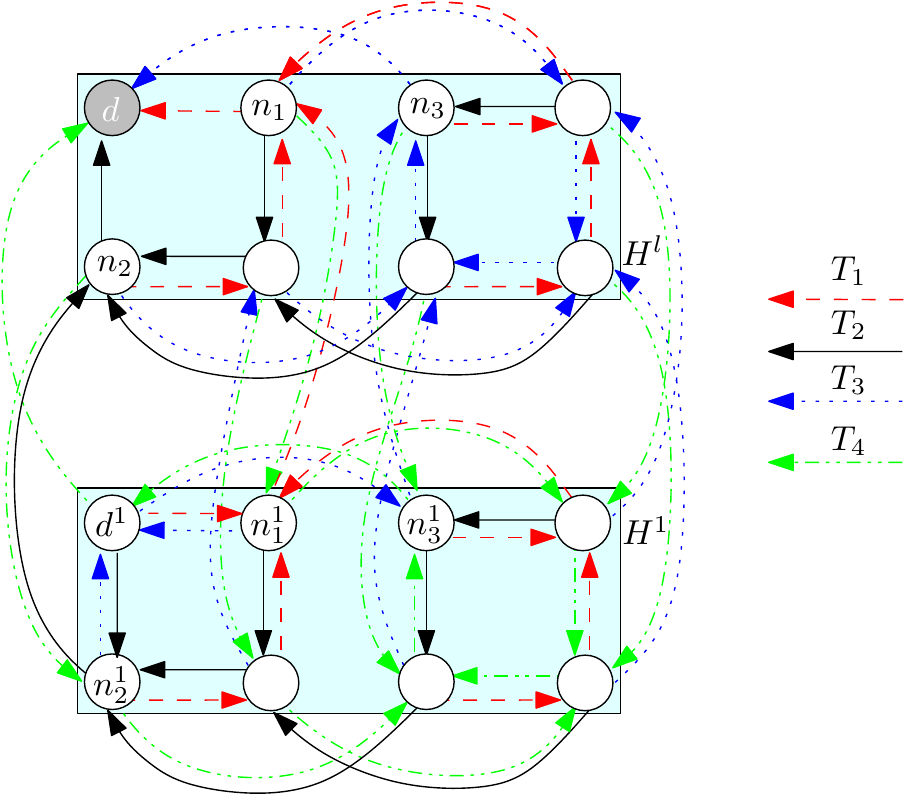}
    \caption{A $(4,1)$-generalized hypercube with four arc-disjoint arborescences (solid black, dashed red, dotted blue, and dash dotted green). }
\label{fig:hypercube-partitioning-induction-16}
\end{figure}

 Observe that, by construction, when a packet is routed through $T_{k^i+l}$ in $H(i,k,l+1)$, it is never rerouted through any other arborescence. In fact, rerouting on different arborescences, happens only when a packet is sent from $H^1$ to $H^l$, unless the destination vertex is $d$. Since $T_{k^i+l}$ does not include any edge from $H^1$ to $H^l$, except  $(d^1,d)$, and it is build from $T_{k^i+l-1}^l$, which, by induction hypothesis, does not reroute on $T_1^l$, the statement easily follows.
 
 We now prove that this routing function for $H^{l+1}$ is $(k^i+l)$-\nameprop such that if a packet is originated at a vertex $n_i$, then it is routed to a vertex $n_j$, with $j>i$, only when it is routed along $T_j$.    
  Consider a packet $p$ that is originated in $H^l$. Observe that all the arcs from $H^l$ to $H^1$ belong to $T_{k^i+l}$ and, since routing in $H^l$ is $(k^i+l-1)$-\nameprop, a vertex of $H^l$ can detect when  $(k^i+l-1)$ edges failed.
  In that case, a packet is forwarded through the next arborescence, i.e., $T_{k^i+l}$, and it is routed entirely within $H^1$ plus $(d^1,d)$. Since a packet $p$ routed through $T_{k^i+l}$ is never rerouted to any other arborescence, we have that either $p$ reaches $d^1$ and, in turn, $d$  or an edge in $H^1$ or between $H^l$ and $H^1$ must have failed. 
  This edge is different from  any other of the $(k^i+l-1)$ edges that failed in $H^l$, which leads to a total of $k^i+l$ edge failures. The vertex that cannot forward through $T_{k^i+l}$ detects that at least $k^i+l$ edges failed. 
  Hence, in this case, the routing function is $(k^i+l)$-\nameprop since a packet will never enter a loop without any vertex detecting that $k^i+l$ edges failed.
  

 Before considering a packet that is originated from $H^1$, we prove that if a packet is routed at a vertex $n_i$ along $T_i$, then it does not traverse any vertex $n_1,\dots,n_{i-1}$ while it is routed through $T_i,\dots,T_k$.  Observe that vertices $n_1,\dots,n_{k^i+l-1}$ are all contained in $V(H^l)$ and a packet is routed to $H^1$ (and in turn to $n_{k^i+l}$), only along $T_{k^i+l}$. Hence, by induction hypothesis and since a packet routed along $T_{k^i+l}$ is never rerouted to $T_1$, this property holds.

  We now consider a packet $p$ that is originated from a vertex of $H^1$. Observe that since the routing function within $H^1$ have been partially modified, we first need to analyze these differences. 
  These changes in the routing functions only involve $d^1$ and its neighbors in $H^1$ and the fact that all vertices will route from $T_{k^i-1}$ through a set of arborescences $T_{k^i},\dots,T_{k^i+l-2}$ that connects each vertex of $H^1$ with a direct edge to a vertex of $H^l$. If a packet reaches $H^l$, then it is rerouted along $T_1$ and we already prove that it is guaranteed to reach $d$. Otherwise, it is forwarded through $T_{k^i+l-1}$ and, alternatively, through $T_{k^i+l}$. We first consider a packet $p$ that is not originated at $d^1$. In this case, observe that a packet $p$ is routed through $T_1,\dots,T_{k^i-2}$ exactly as it was routed through $T_1^1,\dots,T_{k^i-1}^1$, with a small difference: If $p$ reaches a vertex $n_j^1$, while it is routed along $T_j$, with $j=1,\dots,k^i-1$, instead of being routed to $d^1$, it is routed through  arc $(n_j^1,n_{j,1})$. If that edge is failed, $p$ is routed exactly as if edge $\{n_j^1,d^1\}$ failed in $T_j^1$. Hence, a packet either reaches a vertex of $H^l$ or it is routed according to $T_1^i,\dots,T_{k^i-1}^1$. Now, if a packet hits a failed edge along $T_{k^i-1}$, it is rerouted along the next $l-1$ arborescences $T_{k^i},\dots,T_{k^i+l-2}$, which directly connect each vertex of $H^1$ to a vertex in $H^l$. In that case, a packet $p$ either reaches a vertex in $H^l$, for which we are guaranteed that it will reach $d$ or a vertex detects that $k^i+l$ edges failed,  or $k^i+l-1$ distinct edges failed and $p$ is routed inside $H^1$ along  $T_{k^i+l-1}$. In the latter case, $p$ is either routed to a vertex in $H^l$, or it hits a failed edge $e$ that, by construction of $T_{k^i+l-1}$, either connects $H^1$ to $H^l$ or it is incident to $d^1$. 
  
  \eat{
  In fact, by inductive hypothesis, when $p$ is rerouted along $T_{k^i+l-2}$, it is not at any vertex $n_j^1$, with $j<k^i+l-3$, so the only edge incident to $d$ is $(n_{k^i-1}^1,d^1)$. Hence, either an edge from $H^1$ to $H^l$ failed or edge $(n_{k^i-1}^1,d^1)$ failed. In both cases, this is a distinct }
  Observe that $e$ is a distinct failed edges. In fact,  all the edges failed along $T_1,\dots,T_{k^i+l-2}$ are not incident to $d^1$. If $p$ hits $e$ and it is rerouted along $T_{k^i+l}$, it cannot hit $e$ in the opposite direction since the outgoing edge at $d^1$ in $T_{k^i+l}$ is towards $d$. Hence, by induction hypothesis and since $T_{k^i+l-1}$ does only route a packet either directly to a vertex of $H^l$, to $d$ or one of its neighbors, when $p$ is rerouted along $T_{k^i+l}$, it is guaranteed to either reach $d^1$, and in turn $d$, or to hit the $k^i+l$ distinct failed edge, which proves the statement of theorem in this case. 
 We now finish our proof by studying how a packet $p$ that is originated at $d^1$ is routed in $H^1$. If all edges incident to $d^1$ failed, we have that $k^i+l$ distinct edges failed and $d$ can detect it. Otherwise, if not all these edges failed, then $p$ is routed to a vertex $n_j^1$, with $j=1,\dots,n_{k^i}$. After that, by inductive hypothesis, we have that packet $p$ is guaranteed to do not traverse any vertex $n_h^1$, with $h<j$. Hence, it is either routed to a vertex in $H^l$ through an arborescence in $T_j,\dots,T_{k^i+l-1}$ or it is routed along $T_{k^i+l}$.  
  In that case, $p$ may be routed along  $T_{k^i+l-1}$ for at least an edge or not. In the first case, by construction of $T_{k^i+l-1}$, $p$ is at $d$ or $n_{k^i}^1$. In both cases, it can be routed to $d$ along $T_{k^i+l}$ and if $(d^1,d)$ is failed, a vertex can detects that $k^i+l$ distinct edges failed, which proves the statement of the theorem. In the second case, a packet does not change its location if an edge failed along $T_{k^i+l-1}$. Hence, by inductive hypothesis, it is guaranteed to be routed to $(d^1,d)$ or to hit a distinct failed edge. Hence, the statement of the theorem is proved in this case as well.
\end{proof}

\eat{

 \noindent\textbf{More results on hypercubes. } The longest path of an arborescence is $O(\lg{n})$, where $n$ is the number of vertices
 in the arborescence. This is much better than technique that uses hamiltonian cycles, where the longest path is $O(n)$.
 \textbf{Open problem:} What is the maximum strech?
}

\subsection{Clos networks} A $k$\emph{-Clos network}~\cite{fattree} is a $k$-connected graph
 that consists of $k$ partially overlapping multirooted trees organized in layers. 
 Its high bisection bandwidth and symmetric structure make it an ideal choice for a datacenter network topology.
 Our \nameprop routing function construction decomposes a Clos network into 
 a set of $k$-connected complete bipartite graphs, where each vertex belongs to at most two complete bipartite graphs. We first compute a $k$-\nameprop routing function for each of the $k$-connected complete bipartite graph. After that, we interconnect all these 
 bipartite graphs in such a way that the resiliency of the resulting graph is also $k-1$.
 This technique improves upon all previously known results about resiliency 
 in Clos networks~\cite{lhka-f10aften-13} in two ways: First, in our case all the vertices are 
 $(k-1)$-resilient and not only the leaves of the multirooted a $k$-\nameprop routing function; Second, our construction works for any arbitrary number of layers of the Clos network.


\begin{theorem}\label{theo:clos-resiliency}
\ClosResiliency
\end{theorem}

\begin{proof}
 A $k$-connected Clos network $C$ can be decomposed into a tree $T$ such that: (i) a node of $T$ represent a complete bipartite subgraph of $C$, (ii) there exists a directed arc from a node $x$ of $T$ to a node $y$ of $T$ if $y$ contains a vertex of $C$ that is closer to $d$  than any vertex of $C$ contained in $x$, and (iii) each vertex of $C$ belongs to at least one node (at most two nodes) of $T$. Let $n_1,\dots,n_l$ be the set of nodes of $T$. 
 Let $n_1$ be a complete bipartite graph that contains $d$ and for each graph $n_i$, with $i=2,\dots,l$ let $d_i$ be an arbitrary vertex of $n_i$ that is closer to $d$. By Theorem~\ref{theo:complete-bipartite-resiliency}, we can construct within each $n_i$, a $k$-\nameprop routing function towards $d_i$. When a packet reaches $d_i$, it is routed through the next complete bipartite graph $n_j$, with $j\neq i$, towards a destination $d_j$ that is closer to $d$ than $d_i$. Since, we are using a shortest path metric, such destination must exists. Hence, the statement of the theorem is proved.
\end{proof}

\subsection{Two dimensional grids (Hamiltonian-based routing)}
 A $2$-dimensional $n\times m$ \emph{grid} consists of $n+m$ cycles $c_1,\dots,c_n,c_1',\dots,c_m'$, where $c_i=(v_{1i},\dots,v_{mi})$ and $c'_i=(v_{i1},\dots,v_{in})$.
 %
%
 We now introduce a useful technique based on Hamiltonian cycles that can be used to construct $(k-1)$-resilient $k$-\nameprop routing functions.
 Consider a sequence $S$ of $2k$ arc-disjoint arborescences $S=<T_1^A,T_1^B,\dots,T_k^A,T_k^B>$ 
 where $T_i^A$ is a path $(v,w_1,\dots,w_n,u,d)$ and $T_i^B=(u,w_n,\dots,w_1,v,d)$ 
 is the same path reversed. Now, if routing functions route packets
 according to this ordered sequence $S$, we obtain $(k-1)$-resiliency. 
 In fact, when a packet hits a failed edge on a path $T_i^A$, it is sent
 in the opposite direction, where it is guaranteed to either reach $d$ or to hit a different failed edge. Since any arborescence
 $T_i^A$ or $T_i^B$ does not overlap with any other arborescence $T_j^A$ and $T_j^B$, with $j\neq i$,
 this is a set of $(2k-1)$-\nameprop routing functions, 
 Observe that paths $T_i^A$ and $T_i^B$ form a Hamilatonian cycle,
 i.e., a cycle that visits all vertices exactly once. 
 Hence, if a graph contains $k$ edge-disjoint Hamiltonian cycles, 
 then we can exploit these cycles to easily construct $(2k-1)$-resilient routing functions.
 This allows us to exploit known results about the number of edge-disjoint Hamiltonian
 cycles in specific graphs in order to provide resiliency guarantees. For instance, 
 it is well-known that a $(2i,1)$-generalized hypercube (i.e., a ``standard'' hypercube)
 contains $i$ edge-disjoint Hamiltonian cycles~\cite{abs-hamiltonian-90}, which can be used to compute 
 $(2i-1)$-resilient routing functions. As for grids, we now show  how to compute $2$ edge-disjoint 
 Hamiltonian cycles inside a grid.


\begin{theorem}\label{theo:grid-resiliency}
\GridResiliency
\end{theorem}

\begin{proof}
Our routing scheme relies on a grid graph 
decomposition into $2$ edge-disjoint 
Hamiltonian cycles. We prove that such decomposition always 
exists. We provide patterns for different 
parity of grid dimensions which are extendible 
by adding two rows or columns for such 
decomposition. Fig.~\ref{fig:grid1}, Fig.~\ref{fig:grid2}, and Fig.~\ref{fig:grid3} shows all 
possible parity cases. Two cycles are marked 
with different line types. Repeatable blocks 
are highlighted with curve brackets.
%
%
%
%
%
\begin{figure}[t]
\begin{minipage}[t]{0.25\columnwidth}%
\includegraphics[width=\columnwidth]{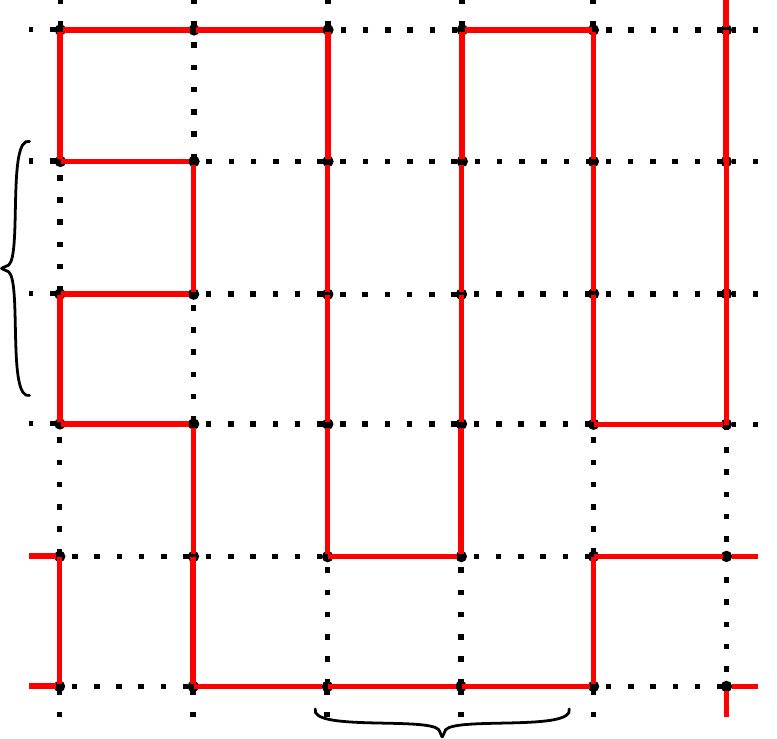}
\caption{Odd-Odd case.}\label{fig:grid1}
\end{minipage}%
\hfill
\begin{minipage}[t]{0.25\columnwidth}%
\includegraphics[width=\columnwidth]{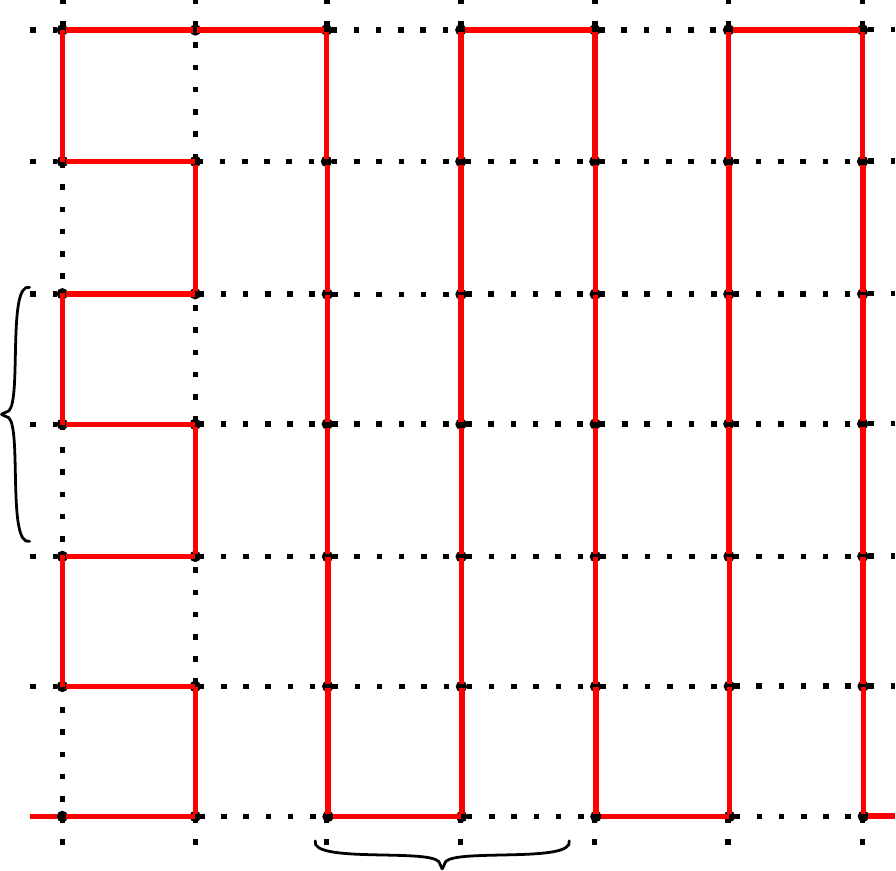}
\caption{Even-Even case.}\label{fig:grid2}
\end{minipage}
\hfill
\begin{minipage}[t]{0.25\columnwidth}%
\includegraphics[width=\columnwidth]{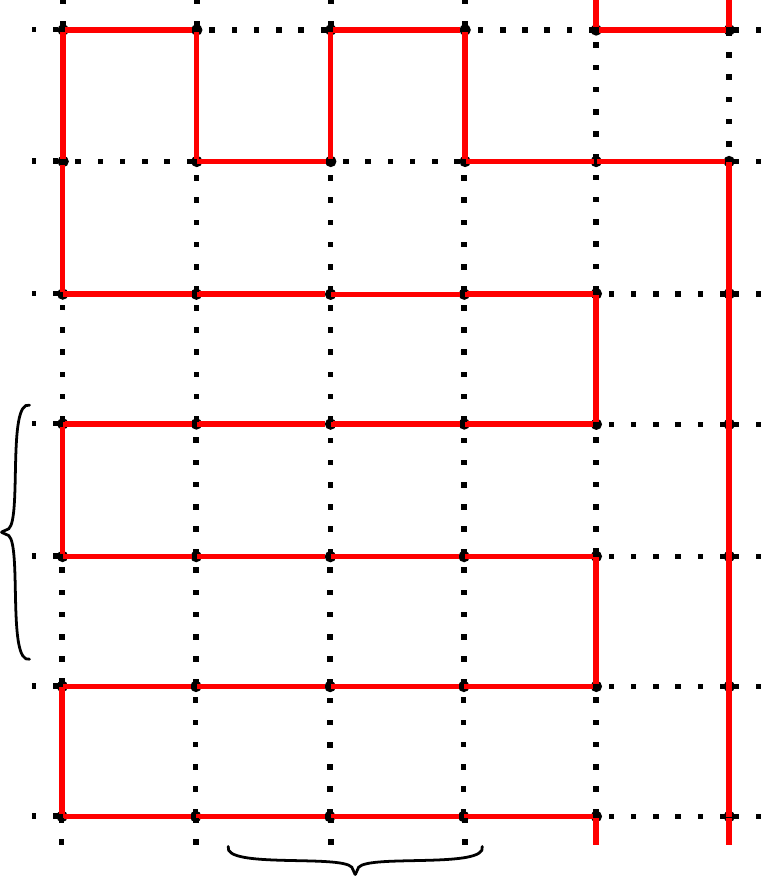}
\caption{Even-Odd case}\label{fig:grid3}
\end{minipage}%
\end{figure}
\end{proof}

\section{Impossibility Results}\label{appe:negative-results}

\vspace{.1in}
\noindent\textbf{Impossibility results. }
\vspace{2mm}
\rephrase{Theorem}{\ref{theo:no-circular-ordering}}{
\noCircularOrdering 
}
\vspace{2mm}

\begin{proof}
\begin{figure}[tb]
  \centering
    \includegraphics[width=.6\columnwidth]{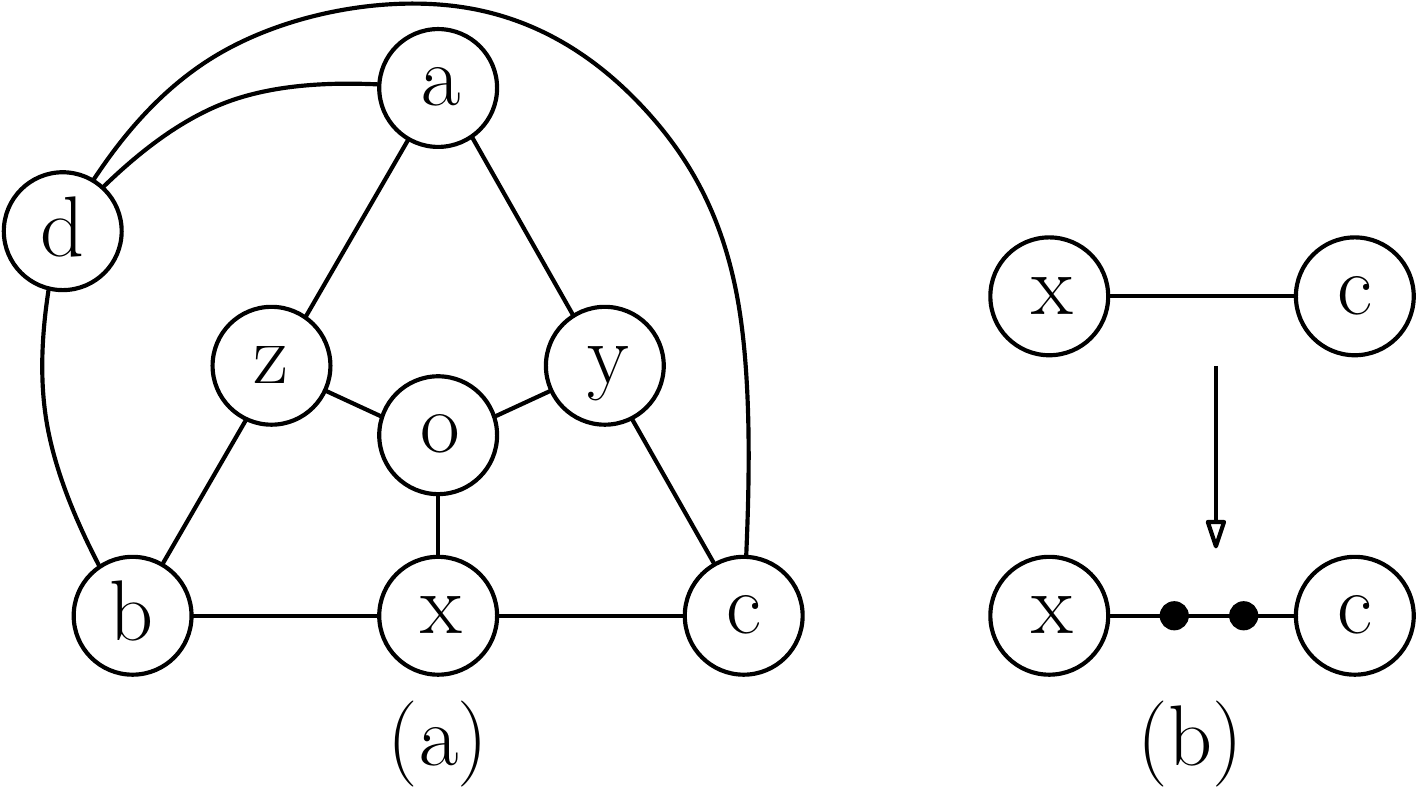}
    \caption{ (a) No circular routing functions can guarantee $2$-resiliency. (b) Edge transformation.}
\label{fig:candidate-counterexample-ordering-3edges}
\end{figure}
Consider the $3$-connected graph shown in Fig.~\ref{fig:candidate-counterexample-ordering-3edges}(a), where $d$ is the destination. 
Suppose, by contradiction, that there exists a $2$-resilient set of circular routing functions. 
Since the graph is symmetric, w.l.o.g, assume that $o$ routes clockwise, i.e., a 
packet received from $x$ is sent to $z$, from $z$ to $y$, and from $y$ to $x$. Also, w.l.o.g, $o$ sends its originated packet $p$ to $y$ when none of its incident edges fail. 

We first claim that vertices $y$, $a$, and $z$ route counterclockwise. 
Suppose, by contradiction, that (i) $y$ routes clockwise, or (ii) 
$a$ routes clockwise, or (iii) $z$ routes clockwise.
For each case, consider the following failure scenarios.
In case (i), suppose both edges $(a,d)$ and $(z,b)$ fail. In case (ii), suppose both edges $(y,c)$ and $(z,b)$ fail. In case (iii), suppose both edges $(y,c)$ and $(a,d)$ fail. In each case packet $p$ is routed along $(y,a,z,o,y)$ and a forwarding loop arises---a contradiction.

Observe now that, in the absence of failures, if $c$ sends a packet $p$ to $x$, if $x$ routes clockwise it forwards it directly to $b$, 
otherwise, if $x$ routes counterclockwise, $p$ is forwarded
through $o$, $z$, $a$, $y$, $o$, $x$, and, also in this case, to $b$. 
Consider the scenario where both edges $(c,d)$ and $(b,d)$ failed. A packet $p$ received by $y$ from $o$ is routed from $c$ to $x$ and, because of the previous observation, to $b$. After that, it is routed through $(z,o,y)$ and a forwarding loop arises---a contradiction.
\end{proof}

\vspace{2mm}
\rephrase{Lemma}{\ref{lemm:from-vertex-circular-to-vertex-connectivity}}{
\LemmaFromVertexCircularToVertexConnectivity
}
\vspace{2mm}
\begin{proof}
Replace each edge of $G$ with a path consisting of three edges, as shown in Fig.~\ref{fig:candidate-counterexample-ordering-3edges}(b). 
We call the new added vertices \emph{intermediate} vertices (depicted as small black circles) and the old ones \emph{original} vertices.
Each original vertex of $G$ retains its $3$-connectivity to $d$. 
It is easy to see that intermediate vertices must forward a packet received through one edge to the other one, if it did not fail. Otherwise, if an intermediate vertex $v$ bounces back to a vertex $u$ a packet, then if all edges incident at $u$ fail, except $(v,u)$, a forwarding loop arises. 
This implies that we only need to compute routing functions at original vertices. 

We now prove that the routing functions at the $8$ original vertices, except $d$, must be \VertexCircular.
Once we prove this, the statement of the theorem easily follows from Theorem~\ref{theo:no-circular-ordering}, where we proved that no \VertexCircular routing functions can guarantee $2$-resiliency on $G$.
From now on, we will consider only failures between two intermediate vertices, thus a routing table at each original vertex consists of just four entries: 
Where to send a packet received from each of its three neighbors $n_1$, $n_2$, and $n_3$ and where to send a locally originated packet. We can discard the last entry as it does not influence if a routing table is circular.
Hence, we simplify our routing table notation as follows. Let $f^v(n)=n'$ be a routing table at vertex $v$ such that a packet received from a neighbor $n$ is  forwarded to a neighbor $n'$.

We make the following observations.
First, for each original vertex $v$, we have that $f^v(n) \ne n$, with $n \in \{n_1,n_2,n_3\}$ i.e. no vertex bounces a packet back to the edge where it received it, exactly as in the case of intermediate vertices. 
Second, all entries in the routing table are distinct. Otherwise, suppose by contradiction that, w.l.o.g., 
$f^v(n_1) = f^v(n_2) = n_3$ and $f^v(n_3)=n_1$.
If both $n_1$ and $n_3$ have a dead-end ahead because of two edge failures, then a forwarding loop among $n_3$, $v$, and $n_1$ arises. Hence, the routing function at each vertex must be \VertexCircular. Since a \VertexCircular routing function at intermediate vertices consists in forwarding a packet to the other edge, it easily follows that the same \VertexCircular routing functions at original vertices are $2$-resilient for $G$.
\end{proof}

\vspace{2mm}
\rephrase{Theorem}{\ref{theo:vertex-connectivity-impossibility}}{
\vertexConnectivityTheorem
}
\vspace{2mm}

\begin{proof}
Consider the graph $G$ used in the proof of Theorem~\ref{theo:no-circular-ordering} in Fig.~\ref{fig:candidate-counterexample-ordering-3edges}(a). 
By Lemma~\ref{lemm:from-vertex-circular-to-vertex-connectivity}, $G'$ must implement a \VertexCircular routing function. This is in contradiction
with Theorem~\ref{theo:no-circular-ordering}, which states that $G$ does not allow any
\VertexCircular routing function. 
\end{proof}

\vspace{2mm}
\rephrase{Theorem}{\ref{theo:no-k-resiliency}}{
\noKResiliency 
}
\vspace{2mm}

\begin{proof}
Consider the graph $G$ used in the poof of Theorem~\ref{theo:vertex-connectivity-impossibility}. After having applied all the edge transformations, $G$ becomes $2$-connected and we proved that $2$-resiliency cannot be achieved. This implies the statement of the theorem.
\end{proof}

\vspace{2mm}
\rephrase{Theorem}{\ref{theo:no-k-resiliency}}{
\noKResiliency 
}
\vspace{2mm}

\section{Randomized Routing}\label{appe:probabilistic}

\vspace{2mm}
\rephrase{Lemma}{\ref{lemma:tree-components}}{
\TreeComponentsLemma
}
\vspace{2mm}

	\begin{proof}
		We give a proof by contradiction. To that end, assume that the set of connected components of $\HF$, denoted by $\cC$, contains at most $k - f - 1$ trees. Now, if $C \in \cC$ is a tree, we have $|E(C)| = |V(C)| - 1$, and $|E(C)| \ge |V(C)|$ otherwise. We also have
		\begin{eqnarray}\label{eq:non-tree-tree}
			\sum_{C \in \cC}{|E(C)|} & = & \sum_{C \in \cC \text{ is not a tree}}{|E(C)|} + \sum_{C \in \cC \text{ is a tree}}{|E(C)|} \nonumber \\
				 & \ge & \sum_{C \in \cC \text{ is not a tree}}{|V(C)|} + \sum_{C \in \cC \text{ is a tree}}{(|V(C)| - 1)}. 
		\end{eqnarray}
		Next, following our assumption that $\cC$ contains at most $k - f - 1$ trees, from \eqref{eq:non-tree-tree} we obtain
		\begin{equation}\label{eq:bound-on-all-comp}
			\sum_{C \in \cC}{|E(C)|} \ge \sum_{C \in \cC}{|V(C)|} - (k - f - 1).
		\end{equation}
		Furthermore, as by the construction we have $\sum_{C \in \cC}{|V(C)|} = |V_F| = k$, \eqref{eq:bound-on-all-comp} implies
		\begin{equation}\label{eq:final-bound-on-E}
			\sum_{C \in \cC}{|E(C)|} \ge |V_F| - (k - f - 1) = f + 1.
		\end{equation}
		On the other hand, from the construction of $\HF$ we have
		\[
			\sum_{C \in \cC}{|E(C)|} = f,
		\]
		which leads to a contradiction with \eqref{eq:final-bound-on-E}.
	\end{proof}

	\vspace{2mm}
    \rephrase{Lemma}{\ref{lemma:good-arcs}}{
        \GoodArcsLemma
    }
    \vspace{2mm}
	
	\begin{proof}
		Let $T_i$ be an arborescence of $\cT$ such that $i \in V(T)$. Then, by the construction of $\HF$ we have that $T_i$ contains a failed link. Next, a failed link closest to the root of $T_i$ is a good arc of $T_i$. Therefore, for every $i \in V(T)$, we have that $T_i$ contains an arc which is both good and failed. Furthermore, by the construction of $\HF$ and the definition of well-bouncing arcs, we have that for every good, failed link of $T_i$ there is the corresponding well-bouncing arc of $\dT$. Also, observe that the construction of $\HF$ implies that a well-bouncing arc corresponds to exactly one good-arc.
		
		Now, putting all the observations together, we have that each $T_i$, for every $i \in V(T)$, has a good failed link which further corresponds to a well-bouncing arc of $\dT$. As all the arborescences are arc-disjoint, and there are $|V(T)|$ many of them represented by the vertices of $T$, we have that $\dT$ contains at least $|V(T)|$ well-bouncing arcs.
	\end{proof}
	
	\vspace{2mm}
    \rephrase{Lemma}{\ref{lemma:at-least-one-good}}{
        \AtLeastOneGoodLemma
    }
    \vspace{2mm}
	
	\begin{proof}
		Consider two cases: $|V(T)| = 1$, and $|V(T)| > 1$. In the case $|V(T)| = 1$, $T$ is an isolated vertex which implies that it has no outgoing arcs. Therefore, $T$ represents a good arborescence.
		
		If $|V(T)| > 1$, then from Lemma \ref{lemma:good-arcs} we have that $\dT$ contains at most $2 (|V(T)| - 1) - |V(T)| < |V(T)|$ arcs which are not well-bouncing. This implies that there is at least one vertex in $T$ from which every outgoing arc is well-bouncing.
	\end{proof}
	
		\vspace{2mm}
    \rephrase{Theorem}{\ref{theorem:main-theorem-probabilistic}}{
        \MainTheoremProbabilistic
    }
    \vspace{2mm}
	
	\begin{proof}
		Let $F$ be the set of failed links, at most $k - 1$ of them. Then, by Lemma~\ref{lemma:tree-components} we have that $\HF$ contains at least $k - f \ge 1$ tree-components. Let $T$ be one such component.
		
		By Lemma~\ref{lemma:at-least-one-good} we have that there exists at least one arborescence $T_i$ such that $i \in V(T)$ and every outgoing arc from $i$ is well-bouncing. Now, as on a failed link algorithm \probAlgo will switch to $T_i$ with positive probability, and on a failed link of $T_i$ the algorithm will bounce with positive probability, we have that the algorithm will reach $d$ with positive probability.
	\end{proof}
	
\vspace{.1in}
\noindent\textbf{Calculations omitted from Section~\ref{sect:running-time-probabilistic}}
		Subtracting \eqref{eq:X} from \eqref{eq:Y} we obtain
		\begin{equation}\label{eq:X-Y}
			\EE{Y} = q \EE{X}.
		\end{equation}
		Substituting \eqref{eq:X-Y} to \eqref{eq:X} gives
		\begin{equation}\label{eq:refined-X}
			\EE{X} = \frac{1}{(1 - q) q (1 - t)},
		\end{equation}
		and therefore, from \eqref{eq:X-Y},
		\begin{equation}\label{eq:refined-Y}
			\EE{Y} = \frac{1}{(1 - q) (1 - t)}.
		\end{equation}
		Substituting \eqref{eq:refined-X} and \eqref{eq:refined-Y} into \eqref{eq:I}, we obtain an upper bound on $\EE{I}$
		$$
			 \ \ \ \ \ \ \ \ \ \   \ \ \ \ \ \ \  \  \ \ \ \ \ \ \  \EE{I}  \le  \frac{t}{(1 - q) q (1 - t)} + \frac{1}{1 - q}. \ \ \ \ \ \ \ \ \ \  \ \ \ \ \ \ \ \ \ (\ref{eq:I-upper-bound})
		$$

	Let $U(q)$ denote the upper-bound provided by \eqref{eq:I-upper-bound}, i.e.
		\begin{equation}\label{eq:U}
			U(q) := \frac{t}{(1 - q) q (1 - t)} + \frac{1}{1 - q}.
		\end{equation}
		Now we can prove the following lemma.
		
		\newcommand{\ProbabilisticRunningTimeOne}{It holds
			\[
				\EE{I} \le 2 + 4 \frac{t}{1 - t}.
			\]}

		\begin{lemma}\label{lemma:probabilistic-running-time-one}
			\ProbabilisticRunningTimeOne
		\end{lemma}
	
    \begin{proof}
		From \eqref{eq:I-upper-bound} we have $\EE{I} \le U(q)$. Setting $q = 1/2$ in \eqref{eq:U} we obtain
		\[
			U (1/2) \le 2 + 4 \frac{t}{1 - t},
		\]
		and the lemma follows.
	\end{proof}

\vspace{2mm}
    \rephrase{Lemma}{\ref{lemma:probabilistic-running-time-two}}{
        $U(q)$ is minimized for
			\[
				q = q^* := 1 - \frac{1}{1 + \sqrt{t}},
			\]
			and equal to
			$$
				 \ \ \ \ \ \ \ \ \ \ \ \ \ \ \ \ \ \ \ \ \ \ \ \ \ \ \ \ \ \ \ \ \ \ \ \ \ \  
				 U(q^*) = \frac{1 + \sqrt{t}}{1 - \sqrt{t}}. 
				 \ \ \ \ \ \ \ \ \ \ \ \ \ \ \ \ \ \ \ \ \ \ \ \ \ \ \ \ \ \ \ \ \ \ \ \ 
				 (\ref{eq:opt-U})
			$$
    }\vspace{2mm}
 		\begin{proof}
			Consider $U(q)'$, which is
			\[
				U(q)' = \frac{t (1 - q)^2 - q^2}{(1 - q)^2 q^2 (t - 1)}.
			\]
			In order to find the value of $q$ that minimizes $U(q)$, denote it by $q^*$, we find the roots of $U(q)' = 0$ with respect to $q$. There is only one positive solution of equation $U(q)' = 0$, which is also the minimizer $q^*$, and is equal to
			\[
				q^* = 1 - \frac{1}{1 + \sqrt{t}},
			\]
			as desired.
			
			Finally, substituting $q^*$ into \eqref{eq:U} and simplifying the expression we obtain \eqref{eq:opt-U}.
		\end{proof}

\vspace{2mm}
    \rephrase{Theorem}{\ref{theo:never-bounce}}{
        \NeverBounce
    }\vspace{2mm}
    
\begin{proof}
 We now define a $2k$ edge connected graph $G$ and a set of $2k$ arc disjoint spanning trees $T_1,\dots,T_{2k}$ as follows. The set of vertices $V$ of $G$ consists of a destination vertex $d$ and $4k$ additional vertices arranged into three equal-sized layers $L_1=\{v_1^1,\dots,v_{2k}^1\}$ and $L_2=\{v_1^2,\dots,v_{2k}^2\}$. 
 Edges are added in such a way that $L_2$ is a clique of size $2k$ and $(L_1,L_2)$ is a complete bipartite graph.
 Vertex $d$ is connected to each vertex in $L_1$.
 We now show how to construct $2k$ arc-disjoint spanning trees $T_1,\dots,T_{2k}$.
 For each $i=1,\dots,k$, add into $T_{2i}$ arcs $(v_{2i}^2,v_{2i}^1)$, $(v_{2i}^1,v_{2i+1}^1)$, $(v_{2i+1}^1,d)$ and add into $T_{2i+1}$ arcs  $(v_{2i+1}^2,v_{2i+1}^1)$, $(v_{2i+1}^1,v_{2i}^1)$, $(v_{2i}^1,d)$. 
 For each $i=1,\dots,k$, for each $j=1,\dots,2i-1,2i+2,\dots,2k$, add into $T_{2i}$ arcs $(v_j^2,v_{2i+1}^2)$ and $(v_j^1,v_{2i+1}^2)$ and add into $T_{2i+1}$ arcs $(v_j^2,v_{2i}^2)$ and $(v_j^1,v_{2i}^2)$.
 We now consider the failure scenario in which edges $(v_0^2,v_1^2)$, $(v_2^2,v_2^2)$, \dots , $(v_{2k-3}^2,v_{2k-2}^2)$ and $(v_0^1,v_1^1)$, $(v_2^1,v_2^1)$, \dots , $(v_{2k-3}^1,v_{2k-2}^1)$, $(v_{2k-1}^1,v_{2k}^1)$. 
 Consider a packet $p$ that is received by a vertex $v_{2i}^1$ from $v_{2i}^2$, which means that $p$ is forwarded through $T_{2i}$. Since edge $(v_{2i}^1,v_{2i+1}^1)$, which belongs to $T_{2i}$, is failed, the only available trees are $T_1,\dots,T_{2i-1},T_{2i+1},\dots,T_{2k}$. Among them only $T_{2i+1}$ has a path that does not contain any failed link from $v_{2i}^1$ to the destination. Every other tree $T_{j}$, connects $v_{2i}^1$ to a vertex $v_j^2$ in $L_2$. Hence the expected number of tree switches $E_1$ when a packet received from a vertex in $L_2$ is routed by a vertex in $L_1$ is $E_1 = \frac{2k-2}{2k-1}E_2+1$, where $E_2$ is the number of expected tree switches when a packet is routed from a vertex $v_i^2$ in $L_2$ along $T_i$. We now compute $E_2$. 
 Consider a packet received by a vertex $v_i^2$ through an edge $(v_i^2,v_j^1)$. By construction of $T_{i}$, $p$ is forwarded along $T_i$. In addition, the outgoing edge of $T_i$ at $v_i^2$ is failed. Hence, $p$ has a probability of $\frac{1}{k-1}$ of being forwarded along $(v_{i}^2,v_i^1)$ and a probability of $\frac{k-2}{k-1}$ of being routed through any other tree 
  $T_j\in \{T_1,\dots,T_{2\lfloor\frac{i}{2}\rfloor-1},T_{2\lfloor\frac{i}{2}\rfloor+2},\dots,T_{2k}\}$
  to vertex $v_j^2$ in $L_2$. Hence, $E_2 = \frac{1}{k-1}E_1 + \frac{k-2}{k-1} + 1$. This leads to $E_1 = (k-1)^2=O(k^2)$.
\end{proof}

We also provide a slightly more involved construction than the one in Theorem~\ref{theo:never-bounce} that shows that there are examples for which if we apply only bouncing, in addition to the number of hops, they have big stretch.

\newcommand{\NeverBounceUpdate}{For any $k>0$, there exists a $2k$ edge-connected graph on $O(N)$ vertices and $O(k^2 + k N)$ edges, a set of $2k$ arc-disjoint spanning trees, and a set of $k-1$ failed edges, such that the expected number of tree switches with \RandomRerouting is $\Omega(k^2)$. Furthermore, the routing makes $\Omega(k^2 N)$ hops in expectation.}

\begin{figure}[t]
\centering
\includegraphics[width=\linewidth/2]{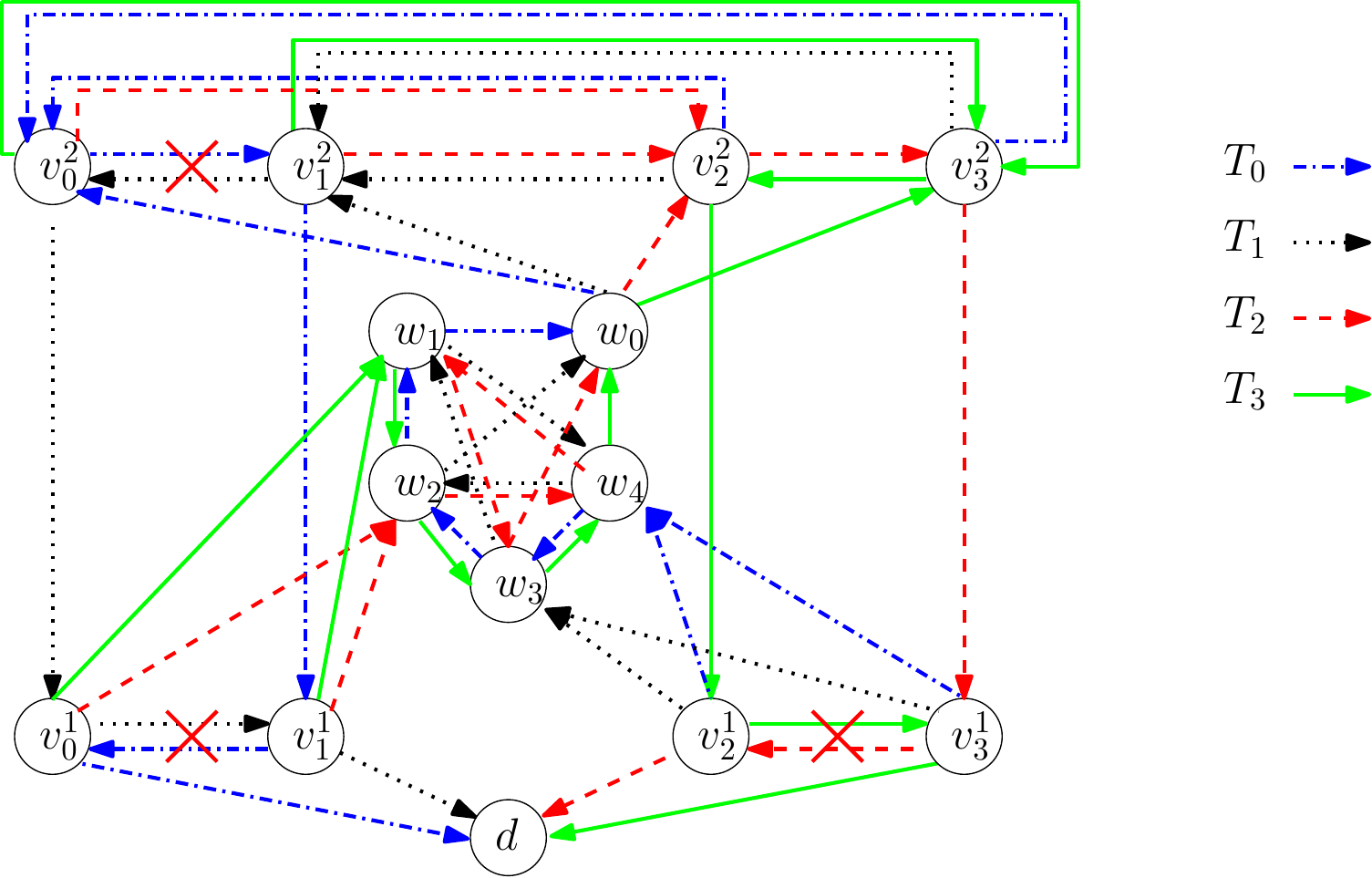}
\captionof{figure}{Graph used in the proof of Theorem~\ref{theo:never-bounce-update} for $k=2$ and $N = 5$.}
\label{fig:toy-gadget-never-bounce-update}
\end{figure}

\begin{theorem}\label{theo:never-bounce-update}
 \NeverBounceUpdate
\end{theorem}

\begin{proof}
 To prove the promised bound, we start by defining a $2k$ edge connected graph $G = (V, E)$ and its set of $2k$ arc disjoint spanning trees $T_0,\dots,T_{2k-1}$ as follows.
 \begin{itemize}
    \item Set $V$ consists of a destination vertex $d$ and $4k + p$ additional vertices arranged into three layers $L_1$, $L_2$, and $W$.
    \item Layers $L_1 = \{v_0^1, \dots, v_{2k-1}^1\}$ and $L_2=\{v_0^2,\dots,v_{2k-1}^2\}$ are equal-sized.
    \item Layer $W = \{w_0, w_1, \ldots, w_{p-1}\}$ is placed "in between" $L_1$ and $L_2$. Number $p$ is a prime such that $\max\{N, 2k + 1\} \le p \le 2 \max\{N, 2k + 1\}$. Note that such $p$ always exist.
    \item Set $E$ is defined to be the edge support of the arborescences that we define in the sequel. Other than that, $G$ does not contain additional edges.
 \end{itemize}
 Next, we construct $2k$ arc-disjoint spanning trees $T_0,\dots,T_{2k-1}$ (see Fig.~\ref{fig:toy-gadget-never-bounce-update} for an example with $k=2$ and $N = 5$). We use $\interval{t}$ to denote set $\{0, 1, 2, \ldots, t-1\}$.

 \begin{itemize}
    \item For each $i \in \interval{k}$, add the following arcs:
        \begin{itemize}
            \item $(v_{2i+1}^2,v_{2i}^2)$, $(v_{2i}^2,v_{2i}^1)$, $(v_{2i}^1,v_{2i+1}^1)$, and $(v_{2i+1}^1,d)$ into $T_{2i+1}$;
            \item arcs $(v_{2i}^2,v_{2i+1}^2)$, $(v_{2i+1}^2,v_{2i+1}^1)$, $(v_{2i+1}^1,v_{2i}^1)$, and $(v_{2i}^1,d)$ into $T_{2i}$.
        \end{itemize}
    \item For each $i \in \interval{k}$, and for each $j \in \interval{2 k} \setminus \{2i, 2i + 1\}$, add the following arcs:
        \begin{itemize}
            \item $(v_j^2,v_{2i}^2)$, $(v_j^1,w_{(2i + 1) (p - 1) \mod p})$, and $(w_0,v_{2i}^2)$ into $T_{2i}$;
            \item $(v_j^2,v_{2i+1}^2)$, $(v_j^1,w_{(2i + 2) (p - 1) \mod p})$, and $(w_0,v_{2i+1}^2)$ into $T_{2i+1}$.
        \end{itemize}
    \item For each $i \in \interval{2 k}$ and each $a \in \interval{p - 1}$ add arc $(w_{(i + 1) (a + 1) \mod p}, w_{(i + 1) a \mod p})$ to $T_i$.
 \end{itemize}
 Finally, consider a scenario in which edges $(v_0^2,v_1^2)$, $(v_2^2,v_3^2)$, \dots , $(v_{2k-4}^2,v_{2k-3}^2)$ and $(v_0^1,v_1^1)$, $(v_2^1,v_3^1)$, \dots , $(v_{2k-4}^1,v_{2k-3}^1)$, $(v_{2k-2}^1,v_{2k-1}^1)$ failed.
 
Since $p$ is a prime, it is easy to show that the described arborescences are valid and arc-disjoint.

For each $a = 1, \ldots, 2 k$ the $i$-th vertex of the vertex-cloud in the middle layer arborescence $a$ walks over has index $x_i^a = a \cdot (i - 1) \mod p$. In order to show that this vertex ordering can indeed be part of $2 k$ arc-disjoint arborescences we will prove that: $x_i^a \neq x_j^a$ whenever $i \neq j$; and, $(x^a_i, x^a_{i + 1})$ is different than $(x^b_j, x^b_{j + 1})$ for any $a \neq b$, and any valid $i$ and $j$. These claims follow from the fact that $\gcd(p, i) = 1$, but for completeness we provide short proofs.

Towards a contradiction, assume that $x_i^a = x_j^a$ for some $i \neq j$. Furthermore, by the definition, that implies $a \cdot (i - 1) \equiv a \cdot (j - 1) \mod p$, and hence $a (i - j) \equiv 0 \mod p$. However, as $0 \le a, |i - j| < p$ and $p$ is a prime, $a (i - j)$ is not divisible by $p$ and hence a contradiction.

Again towards a contradiction, assume that $(x^a_i, x^a_{i + 1}) = (x^b_j, x^b_{j + 1})$ for some $a \neq b$, and some valid $i$ and $j$. Then, from $x^a_i = x^b_j$ we have
\begin{equation}\label{eq:ai-bj}
    a i \equiv b j \mod p.
\end{equation}
On the other hand, $x^a_{i + 1} = x^b_{j + 1}$ implies $a (i + 1) \equiv b (j + 1) \mod p$, which can be written as
\begin{equation}\label{eq:ai-bj+1}
    a i + a \equiv b j + b \mod p.
\end{equation}
Putting together \eqref{eq:ai-bj} and \eqref{eq:ai-bj+1} we obtain $a \equiv b \mod p$, which contradicts the fact that $a \neq b$ and $1 \le a, b \le k$.

Observe that whenever packet reaches vertex of $W$, it visits all the vertices of $W$ before leaving that layer. Then, the rest of the proof, i.e. computing the number of expected hops and stretch, is analogous to the proof of Theorem~\ref{theo:never-bounce}.

This concludes the analysis.
\end{proof}

\section{Header-Rewriting Routing}\label{appe:header-rewriting}
    Next we provide a set of procedures that can be used to implement \dfalgo. Method \proc{getTreeIndices}, for a given link and $\cT$, simply returns indices of the arborescences containing $(x, y)$ or $(y, x)$

    \begin{codebox}
    	\Procname{\proc{getTreeIndices}$(\{x, y\})$}
    	\zi \Comment The method assumes that at least one arborescence
    	\zi \Comment of $\cT$ contains $(x, y)$ or $(y, x)$.
    	\li	\If $\exists \ T_i, T_j \in \cT$ s.t. $(x, y) \in E(T_i)$ and $(y, x) \in T_j$
    	\li \Then
    	        $index_L \gets \min\{i, j\}$
    	\li     $index_H \gets \max\{i, j\}$
    	\li \Else
    	\li     $k \gets $ the index s.t. 
    	\zi         \quad $(x,y) \in E(T_k)$ or $(y, x) \in E(T_k)$
    	\li     $index_L \gets k$
    	\li     $index_H \gets k$
            \End
        \li \Return $(index_L, index_H)$
    \end{codebox}

    Consider a link $\{x, y\}$, and assume we are interested which arborescence it represents during the routing process. If $(x, y)$ belongs to $T_i \in \cT$ and $(y, x)$ belong to $T_j \in \cT$, then we use $H$ to distinguish between $T_i$ and $T_j$. Given $H$ and $\{x, y\}$, method \proc{getTreeIndexGivenH} returns $i$ or $j$ depending on the value of $H$.
    \begin{codebox}
    	\Procname{\proc{getTreeIndexGivenH}$(H, \{x, y\})$}
    	\li $(i_L, i_H) \gets \proc{getTreeIndices}(\{x, y\})$
        \li \If $H \isequal 1$
        \li \Then
                \Return $i_H$
        \li \Else
        \li     \Return $i_L$
            \End
    \end{codebox}
    
    Method \proc{getHGivenTreeIndex} is in a sense the inverse of \proc{getTreeIndexGivenH}. Namely, given a tree index $i$ and a link $\{x, y\}$, method \proc{getHGivenTreeIndex}$(i, \{x, y\})$ returns value $H$ such that \proc{getTreeIndexGivenH}$(H, \{x, y\})$ returns $i$.
    \begin{codebox}
    	\Procname{\proc{getHGivenTreeIndex}$(i, \{x, y\})$}
    	\li $(i_L, i_H) \gets \proc{getTreeIndices}(\{x, y\})$
        \li \If $i \isequal i_H$
        \li \Then
                \Return $1$
        \li \Else
        \li     \Return $0$
            \End
    \end{codebox}
    
    Given three bits $RM$ and $H$, and arc $(x, y)$ method \proc{getTreeIndex} returns index $i$ such that $T_i \in \cT$ is the arborescence that the parameters correspond to.
    \begin{codebox}
    	\Procname{\proc{getTreeIndex}$(RM, H, (x, y))$}
        \li \If $RM \isequal 0$
        \li \Then
                \Return $i$ such that $T_i$ contains $(x, y)$
        \li \Else
        \li     \Return \proc{getTreeIndexGivenH}$(H, \{x, y\})$
            \End
    \end{codebox}
    
    Method \proc{getNextArc} returns the next arc the packet should be routed along $T_i$ for given $RM$ and the last arc $(x, y)$ is has been routed along.
    \begin{codebox}
    	\Procname{\proc{getNextArc}$(RM, (x, y), T_i)$}
    	\zi \Comment The method assumes $y \neq d$.
        \li \If $RM \isequal 0$
        \li \Then
                \Return the first arc on $y-d$ path along $T_i$
        \li \ElseIf $RM \isequal 1$
        \li     \Then
                \Return the first arc following $(x, y)$ in $R(T_i)$
        \li \ElseIf $RM \isequal 2$
        \li     \Then
                \Return the first arc following $(x, y)$ in $R^{-1}(T_i)$
            \End
    \end{codebox}
    
    Finally, we put together all the methods to obtain the main routing algorithm. It should be invoked with \proc{Route}$(0, 0, (x, y))$, where $x \neq d$ is a node the routing has started at, and $(x, y)$ is the first arc on $x-d$ path of $T_1$.
    \begin{codebox}
    	\Procname{\proc{Route}$(RM, H, (x, y))$}
    	\li $i \gets $ \proc{getTreeIndex}$(RM, H, (x, y))$
    	\li	\If $\{x, y\}$ is a failed link
    	\li \Then
    	        \If $RM \isequal 0$
    	\li     \Then
    	        $j \gets \proc{getTreeIndexGivenH}(1 - H, \{x, y\})$
    	\li     \If $j \neq i$
    	\li     \Then
    	            $(x, z) \gets \proc{getNextArc}(1, (y, x), T_j)$
    	\li         $H' \gets \proc{getHGivenTreeIndex}(j, \{x, z\})$
    	\li         \proc{Route}$(1, H', (x, z))$
    	\li     \Else \Comment $(y, x) \notin T_k$, for each $T_k \in \cT$
    	\li         \proc{Route}$(2, 0, (x, y))$
    	        \End
    	\li     \ElseIf $RM \isequal 1$
    	\li     \Then
    	            $(x, z) \gets \proc{getNextArc}(2, (y, x), T_i)$
    	\li         $H' \gets \proc{getHGivenTreeIndex}(i, \{x, z\})$
    	\li         \proc{Route}$(2, H', (x, z))$
    	\li     \Else
    	\li         $j \gets \proc{getTreeIndex}(0, 0, (x, y))$
    	\li         $(x, z) \gets$ the first arc on the $x-d$ path of $T_{j + 1}$
    	\li         \proc{Route}$(0, 0, (x, z))$
    	        \End
    	\li \ElseIf $y \isequal d$
    	\li \Then
    	         Move along $(x, y)$ and finish the routing.
    	\li \Else
    	\li     $(y, z) \gets \proc{getNextArc}(RM, (x, y), T_i)$
    	\li     $H' \gets \proc{getHGivenTreeIndex}(i, \{y, z\})$
    	\li     \proc{Route}$(RM, H', (y, z))$
            \End
    \end{codebox}

\vspace{2mm}
\rephrase{Theorem}{\ref{theo:df-algo-resiliency}}{
    \DFAlgoResiliency
}
\vspace{2mm}
\begin{proof}
    Let $T_i$ be an arborescence of $\cT$ defined in Lemma~\ref{lemma:good-arborescence}, i.e., a good arborescence. Then, \dfalgo will either deliver a packet to $d$ before routing along $T_i$ in canonical mode, or it will route the packet along $T_i$ in canonical mode, which is guaranteed by the fact that \CircularRouting routing is used.
    \\
    Now, if the packet is routed along $T_i$ in canonical mode, either the packet will be delivered to $d$ without any interruption, or it will hit a failed edge of $T_i$ and bounce. But then, if it bounces, by Lemma~\ref{lemma:good-arborescence} and our choice of $T_i$ the packet will reach $d$ without any further interruption.
    
    Therefore, in all the cases the packet will reach $d$.
\end{proof}
\section{Duplication Routing}\label{appe:duplication}

\subsection{Even connected case }

First, we consider the case when $k = 2s$ is even.
\vspace{2mm}
\begin{lemma}\label{lemm:duplication-failed-arc-on-each-tree}
\DuplicationFailedArcOnEachTree
\end{lemma}
\begin{proof}
Step~\ref{step3a} guarantees that the algorithm will route the packet along each $T_i$, with $1 \le i \le s$, and, since it fails, each $T_i$, $1 \le i \le s$, must contains an arc that belongs to a failed edge.
Step~\ref{step3b} guarantees that the algorithm will route the packet along each $T_i$, with $k < i \le 2s$, and, since it fails, each $T_i$, $s < i \le 2s$, must contains an arc that belongs to a failed edge.
\end{proof}

\vspace{2mm}
\rephrase{Lemma}{\ref{lemm:duplication-good-links}}{
\DuplicationGoodLinks
}
\vspace{2mm}

\begin{proof}
If the statement would not be true, then the algorithm would route the packet to the destination using Step~\ref{step3a}.
\end{proof}

\vspace{2mm}
\rephrase{Lemma}{\ref{lemm:duplication-there-are-k-failed-links}}{
\DuplicationThereAreKFailedLinks
}
\vspace{2mm}

\begin{proof}
 By Lemma~\ref{lemm:duplication-failed-arc-on-each-tree}, we have that each $T_i$ has a failed edge. This trivially implies that each $T_i$ has a good failed arc. By Lemma~\ref{lemm:duplication-good-links} and since in our construction $T_{s+1},\dots, T_{2s}$ do not share failed edges, we have that failed edges that the algorithm approaches in $T_1,\dots, T_{s}$, $s$ many of them, are disjoint from all the good failed arcs of $T_{s+1}, \dots, T_{2s}$, $s$ many of them, otherwise at least a copy of a packet would reach $d$. This concludes the proof.
\end{proof}

\vspace{2mm}
    \rephrase{Theorem}{\ref{theorem:duplication}}{
        \DuplicationTheorem
    }
\vspace{2mm}
\begin{proof}
    Towards a contradiction, assume that \dupalgo fails to deliver a packet to $d$. Then, by Lemma~\ref{lemm:duplication-there-are-k-failed-links} the underlying network contains at least $2s$ failed links, which contradicts our assumption that there are at most $2s - 1$ of them. Therefore, \dupalgo delivers a packet do $d$ if there are at most $2s - 1$ failed links.
    
    Regarding the number of copies of the packet created by the algorithm we consider two cases: $f < s$, and $f \ge s$. In the first case, as $T_1, \ldots, T_s$ are pairwise edge-disjoint, by the Pigeonhole principle we have that there is an arborescence $T_i$, for $i \le f + 1$, such that $T_i$ does not contain any failed edge. Therefore, when the packet is routed along $T_i$ it will reach $d$ without any interruption. On the other, before the packet is routed along $T_i$ algorithm \dupalgo will make at most $f$ copies of the packet. In fact, each copy of a packet that is created at Step~\ref{step3a} is routed along an arborescence $T_l$, where $l>s \ge 1$ because $T_1, \ldots, T_s$ are pairwise edge-disjoint.
    
    In the former case, i.e. when $f \ge s$, the algorithm might encounter a failed edge while routing the packet on each of the arborescences of $T_1, \ldots, T_s$. In that case, it will create exactly $2s - 1$ copies of the packet. In fact, at Step~\ref{step3c}, the original packet is routed along $T_{s+1}$ and $s-1$ copies of that packet are routed through arborescences $T_{s+2},\dots,T_{2s}$.
\end{proof}

\subsection{Odd connected case }
We now present an algorithm to achieve $2k$-resiliency for any $(2k+1)$-connected graph.

\vspace{.1in}
\noindent\textbf{Algorithm} \dupalgoOdd. Let $G$ be a $(2k+1)$-connected graph. Construct $(2k+1)$ arc-disjoint arborescences $T_1,\dots,T_{2k+1}$ such that arborescences $T_1,\dots,T_k$ ($T_{k+2},\dots,T_{2k+1})$ do not share edges each other. By Lemma~\ref{lemm:bipartite-trees}, such arborescences exist. Consider the following routing algorithm:
\begin{enumerate}
\item $p$ is first routed along $T_1$.
\item $p$ is routed along the same arborescence towards the destination, unless a failed link is hit.
\item if $p$ hits a failed link $(x,y)$ along $T_i$, then:
\begin{enumerate}
\item\label{step3aodd} if $i \le k$: two copies of $p$ are created; one copy is forwarded along $T_{i+1}$; the other one is forwarded along $T_l$, where $T_l$ is the arborescence that contains arc $(y,x)$.
\item\label{step3bodd} if $i = k+1$: $k$ copies of $p$ are created; the $j$'th copy, with $1\le j\le k$, is routed along $T_{k+j+1}$.
\item\label{step3codd} if $i > k$: $p$ is destroyed.
\end{enumerate}
\end{enumerate}

\vspace{.1in}
\noindent\textbf{\dupalgoOdd correctness.}
\newcommand{\DuplicationTheoremOdd}{For any $(2k+1)$-connnected graph, \dupalgo computes $2k$-resilient routing functions.  }
\newcommand{\DuplicationFailedArcOnEachTreeOdd}{Each $T_i$ contains an arc that belongs to a failed link.}
\newcommand{\DuplicationGoodLinksOdd}{Let $e_i$ be a failed link that the algorithm approaches while routing along $T_i$, for $1 \le i \le k+1$. Then, $e_i$ is not a good arc of $A_j$, for any $k+1 < j \le 2k+1$.}
\newcommand{\DuplicationThereAreKFailedLinksOdd}{$T_1,\dots,T_{2k+1}$ contain at least $2k+1$ failed links.}
We prove it by contradiction. To that end, assume that there are at most $2k$ failed links, and that \dupalgoOdd fails to send $p$ to the destination. Then, as for the even case, we can make the following observations, under the assumption that the algorithm fails to send $p$ to the destination. We first observe that a packet is routed along every arborescence, which leads to the following lemma.

\begin{lemma}\label{lemm:duplication-failed-arc-on-each-tree-odd}
\DuplicationFailedArcOnEachTreeOdd
\end{lemma}
\begin{proof}
Step~\ref{step3aodd} guarantees that the algorithm will route $p$ along each $T_i$, with $1 \le i \le k+1$, and, since it fails, each $T_i$, $1 \le i \le k+1$, must contains an arc that belongs to a failed edge.
Step~\ref{step3bodd} guarantees that the algorithm will route $p$ along each $T_i$, with $k+1 < i \le 2k+1$, and, since it fails, each $T_i$, $k+1 < i \le 2k+1$, must contains an arc that belongs to a failed edge.
\end{proof}

We observe that each failed edge hit along the first $k+1$ arborescences cannot be a good arc, otherwise this would mean that at least a copy of a packet will reach $d$.

\begin{lemma}\label{lemm:duplication-good-links-odd}
\DuplicationGoodLinksOdd
\end{lemma}

\begin{proof}
If the statement would not be true, then the algorithm would route $p$ to the destination using Step~\ref{step3aodd}.
\end{proof}

By a counting argument, we can leverage Lemma~\ref{lemm:duplication-failed-arc-on-each-tree-odd} and Lemma~\ref{lemm:duplication-good-links-odd} in order to prove the following crucial lemma, which is a contradiction since at most $2k$ edges failed.
\begin{lemma}\label{lemm:duplication-there-are-k-failed-links-odd}
 \DuplicationThereAreKFailedLinksOdd
 \end{lemma}
 \begin{proof}
 By Lemma~\ref{lemm:duplication-failed-arc-on-each-tree-odd}, we have that each $T_i$ has a failed edge. This trivially implies that each $T_i$ has a good failed arc. By Lemma~\ref{lemm:duplication-good-links-odd} and since in our construction $T_{k+2},\dots, T_{2k+1}$ do not share failed edges, we have that failed edges that the algorithm approaches in $T_1,\dots, T_{k}$, $k$ many of them, are disjoint from all the good failed arcs of $T_{k+2}, \dots, T_{2k+1}$, $k $ many of them. Consider the set of failed edges $E_{k+1}$ that is hit by $p$ while it is routed along $T_{k+1}$. Two cases are possible: (i) at least an edge of $E_{k+1}$ is in common with an edge that is also a good arc for an arborescence in $\{T_{k+2},\dots, T_{2k+1}\}$ or (ii) not. Otherwise, we would have $2k+1$ distinct failed edges---a contradiction. 
 In case (i), we have a contradiction since $p$ would be bounced on some good arc that belongs to an arborescence in $\{T_{k+2},\dots, T_{2k+1}\}$ while it is routed along  $T_{k+1}$.
 In case (ii), since \dupalgoOdd bounces on each of the first $k$ arborescences, then at least one of this arc would be good for $T_{k+1}$---a contradiction.
\end{proof}

\begin{theorem}
\DuplicationTheoremOdd
\end{theorem}

\end{document}